\documentclass[sigconf,edbt]{acmart-edbt2019}
\pdfoutput=1 

\newif\ifpublishing
\publishingtrue

\newif\ifdomain
\domainfalse

\newif\ifregex
\regexfalse

\newif\ifnotationsummary
\notationsummarytrue

\newif\iffullexample
\fullexamplefalse

\newif\ifshrinkmode
\shrinkmodetrue

\newif\ifshowtmpgraphs
\showtmpgraphsfalse

\newif\iftacode
\tacodetrue

\usepackage{booktabs} 

\graphicspath{{figures/}}
\DeclareGraphicsExtensions{.eps,.jpeg,.png}
\usepackage{balance}  
\DeclareMathOperator*{\argmax}{arg\,max} 

\usepackage{graphicx}
\usepackage{amsmath}
\usepackage{subfig}
\usepackage{algorithm}
\usepackage[noend]{algorithmic}
\usepackage{xspace}
\usepackage{hyperref}
\usepackage{makecell}
\usepackage{enumitem}
\usepackage{multirow}
\usepackage{bm}

\setcopyright{rightsretained}

\acmDOI{}

\acmISBN{978-3-89318-081-3}

\acmConference[EDBT 2019]{22nd International Conference on Extending Database Technology (EDBT)}{March 26-29, 2019}{Lisbon, Portugal} 
\acmYear{2019}

\begin{document}
\title{Crowdsourced Truth Discovery in the Presence of Hierarchies for Knowledge Fusion}

\author{Woohwan Jung}
\affiliation{%
	\institution{Seoul National University}
}
\email{whjung@kdd.snu.ac.kr}

\author{Younghoon Kim}
\affiliation{%
	\institution{Hanyang University}
}
\email{yhkim7951@gmail.com}

\author{Kyuseok Shim}
\authornote{Corresponding author}
\affiliation{%
	\institution{Seoul National University}
}
\email{shim@kdd.snu.ac.kr}
\renewcommand{\shortauthors}{}

\begin{abstract}
	%
	Existing works for truth discovery in categorical data usually assume that claimed values are mutually exclusive and only one among them is correct.
	However, many claimed values are not mutually exclusive even for functional predicates due to their hierarchical structures.
	Thus, we need to consider the hierarchical structure to effectively estimate the trustworthiness of the sources and infer the truths.
	We propose a probabilistic model to utilize the hierarchical structures and an inference algorithm to find the truths.
	In addition, in the knowledge fusion, the step of automatically extracting information from unstructured data (e.g., text) generates a lot of false claims. 
	To take advantages of the human cognitive abilities in understanding unstructured data, we utilize crowdsourcing to refine the result of the truth discovery.
	We propose a task assignment algorithm to maximize the accuracy of the inferred truths. 
	The performance study with real-life datasets confirms the effectiveness of our truth inference and task assignment algorithms.
\end{abstract}

%
%

\newcommand{\leqmid}{\!\leq\!}
\newcommand{\eqmid}{\!=\!}
\newcommand{\eqshort}{\mspace{2mu}\mathsf{=}\mspace{2mu}}
\newcommand{\eqtiny}{\mathsf{=}}
\newcommand{\cdotmid}{\!\cdot\!}
\newcommand{\plusmid}{\!+\!}
\newcommand{\minusmid}{\!-\!}
\newcommand{\inmid}{\!\in\!}
\newcommand{\midsize}[1]{\!#1\!}
\newcommand{\avgdist}{\emph{AvgDistance}\xspace}
\newcommand{\acc}{\emph{Accuracy}\xspace}
\newcommand{\accgen}{\emph{GenAccuracy}\xspace}
\newcommand{\accs}{\emph{Accuracies}\xspace}
\newcommand{\accgens}{\emph{GenAccuracies}\xspace}

\newcommand{\birthplaces}{\emph{BirthPlaces}\xspace}
\newcommand{\heritages}{\emph{Heritages}\xspace}
\newcommand{\todomessage}[1]{TODO ***** #1 *****}
\newcommand{\eat}[1]{}
\newcommand{\red}[1]{{#1}}

\newcommand{\minisection}[1]{\vspace{0.05in}{\bf \noindent #1:}}
\newsavebox\CBox
\def\textBF#1{\sbox\CBox{#1}\resizebox{\wd\CBox}{\ht\CBox}{\textbf{#1}}}

\maketitle

\section{Introduction}
\label{intro}

Automatic construction of large-scale knowledge bases is very important for the communities of database and knowledge management.
Knowledge fusion (KF) \cite{KnowledgeFusion} is one of the methods used to automatically construct knowledge bases \red{(a.k.a. knowledge harvesting)}.
It collects the possibly conflicting values of objects from data sources and applies \emph{truth discovery} techniques for resolving the conflicts in the collected values.
Since the values are extracted from unstructured or semi-structured data, the collected information exhibits error-prone behavior.
The goal of the \emph{truth discovery} used in knowledge fusion is to infer the true value of each object from the noisy observed values retrieved from multiple information sources \red{while simultaneously estimating the reliabilities of the sources}.
\red{
Two potential applications of knowledge fusion are web source trustworthiness estimation and data cleaning \cite{dong2015knowledge}.
By utilizing truth discovery algorithms, we can evaluate the quality of web sources and find systematic errors in data curation by analyzing the identified wrong values.
}


\minisection{Truth discovery with hierarchies}
As pointed out in \cite{KnowledgeVault,KnowledgeFusion,TDSurveryKDDExp15}, the extracted values can be hierarchically structured.
In this case, there may be multiple correct values in the hierarchy for an object even for functional predicates and we can utilize them to find the most specific correct value among the candidate values. 
For example, consider the three claimed values of `NY', `Liberty Island' and `LA' about the location of the Statue of Liberty in Table~\ref{expl_records}.
Because Liberty Island is an island in NY, `NY' and `Liberty Island' do not conflict with each other. 
Thus, we can conclude that the Statue of Liberty stands on Liberty Island in NY.

We also observed that many sources provide generalized values in the real-life.
Figure~\ref{fig:spread_data} shows the graph of the generalized accuracy against the accuracy of the sources in the real-life datasets \birthplaces and \heritages used for experiments in Section~\ref{sec:experiment}.
The accuracy and the generalized accuracy of a source are the proportions of the exactly correct values and hierarchically-correct values among all claimed values, respectively.
If a source claims exactly correct values without generalization, it is located at the dotted diagonal line in the graph.
This graph shows that many sources in real-life datasets claim with generalized values and each source has its own tendency of generalization when claiming values. 

Most of the existing methods \cite{zhao2012bayesian,LCA,DOCS,dong2009integrating, dong2012less} simply regard the generalized values of a correct value as incorrect.
Thus, it causes a problem in estimating the reliabilities of sources.
According to \cite{KnowledgeFusion}, 35\% of the false negatives in the data fusion task are produced by ignoring such hierarchical structures. 
Note that there are many publicly available hierarchies such as WordNet \cite{wordnet} and DBpedia \cite{auer2007dbpedia}.
\eat{Since there are many publicly available hierarchies such as WordNet \cite{wordnet} and DBpedia \cite{auer2007dbpedia}, we can easily utilize such available hierarchies to infer the truths from conflicting claimed values.}
Thus, a truth discovery algorithm to incorporate hierarchies is proposed in \cite{asums}.
However, it does not consider the different tendencies of generalization and may lead to the degradation of the accuracy.
Another drawback is that it needs a threshold to control the granularity of the estimated truth.

\eat{
Figure~\ref{fig:spread_data} shows the graph of the generalized accuracy against the accuracy of the sources in the real-life datasets \birthplaces and \heritages used for experiments in Section~\ref{sec:experiment}.
The accuracy and the generalized accuracy of a source are the proportions of the exactly correct values and hierarchically-correct values among all claimed values, respectively.
If a source claims exactly correct values without generalization, it is located at the dotted diagonal line in the graph.
As shown in Figure~\ref{fig:spread_data}, there are many sources who claim hierarchically correct values.
In addition, each source has its own tendency of generalization when claiming values. }

\begin{table}[tb]
	\fontsize{7}{7}\selectfont
	\renewcommand{\arraystretch}{1.3}
	\caption{Locations of tourist attractions}
	\vspace{-0.0in}
	\label{expl_records}
	\centering
	\begin{tabular}{c|c|c}
		\hline
		\bfseries Object & \bfseries Source & \bfseries Claimed value\\
		\hline
		Statue of Liberty & UNESCO & NY \\
		Statue of Liberty & Wikipedia & Liberty Island \\
		Statue of Liberty & Arrangy & LA \\
		Big Ben& Quora & Manchester\\
		Big Ben& tripadvisor & London\\
		\hline
	\end{tabular}
	\vspace{-0.0in}
\end{table}

\begin{figure}[b]
	\vspace{-0.08in}
	\centering
	\includegraphics[width=2.35in]{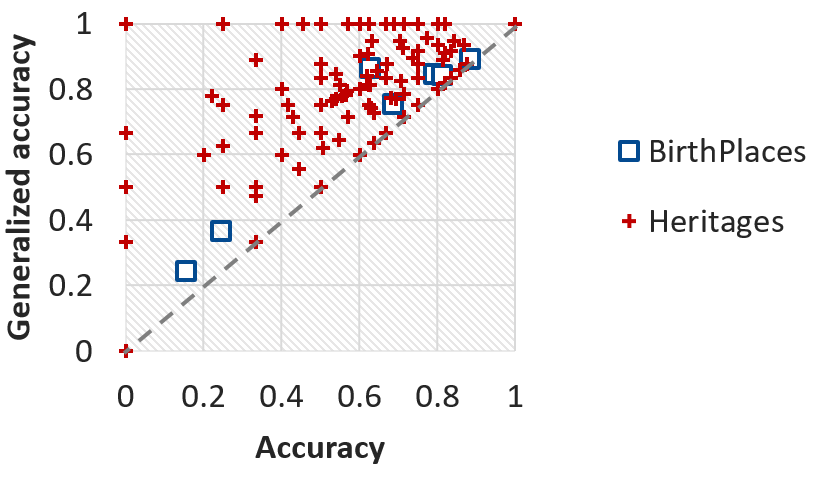}
	\vspace{-0.1in}
	\caption{Generalization tendencies of the sources}
	\label{fig:spread_data}
\end{figure}


\red{We propose a novel probabilistic model to capture the different generalization tendencies shown in \figurename~\ref{fig:spread_data}.
Existing probabilistic models \cite{LCA,DOCS,dong2009integrating, dong2012less} basically assume two interpretations of a claimed value (i.e., correct and incorrect).}
\eat{Thus, the existing works represent the quality of a source by using the probability of providing correct values.}
\red{By introducing three interpretations of a claimed value (i.e., exactly correct, hierarchically correct, and incorrect), our proposed model represents the generalization tendency and reliability of the sources.} 
\minisection{Crowdsourced truth discovery}
\eat{In knowledge fusion, i}It is reported in \cite{KnowledgeFusion} that upto 96\% of the false claims are made by extraction errors rather than by the sources themselves.
Since crowdsourcing is an efficient way to utilize human intelligence with low cost,
it has been successfully applied in various areas of data integration such as schema matching \cite{fan2014hybrid}, entity resolution \cite{wang2012crowder}, graph alignment \cite{kim2017integration} and truth discovery \cite{zheng2015qasca, DOCS}.
Thus, we utilize crowdsourcing to improve the accuracy of the truth discovery.

\red{It is essential in practice to minimize} the cost of crowdsourcing by assigning proper tasks to workers. 
A popular approach for selecting queries in active learning is \emph{uncertainty sampling} \cite{lewis1994sequential, boim2012asking, KimKS17, DOCS}. 
It asks a query to reduce the uncertainty of the confidences on the candidate values the most.
\eat{It is not, however, a proper method for truth discovery since}
However, it considers only the uncertainty regardless of the accuracy improvement.
QASCA algorithm \cite{zheng2015qasca} asks a query with the highest accuracy improvement, 
but measures the improvement without considering the number of collected claimed values.
It can be inaccurate since an additional answer may be less informative for an object which already has many records and answers.

\red{Assume that there are two candidate values of an object with equal confidences.
	If only a few sources provide the claimed values for the object, an additional answer from a crowd worker will significantly change the confidence distribution.
	Meanwhile, if hundreds of sources already provide the claimed values for the object,
	the influence of an additional answer is likely to be very little.
	Thus, we need to consider the number of collected answers as well as the current confidence distribution.}
Based on the observation, we develop a new method to estimate the increase of accuracy more precisely by considering the number of collected records and answers.
We also present an incremental EM algorithm to quickly measure the accuracy improvement and propose a pruning technique to efficiently assign the tasks to workers.

\eat{By combining the proposed task assignment and truth inference algorithms that are superior to the existing methods, our crowdsourced hierarchical truth discovery algorithm is able to infer the truth much more accurately.}



\minisection{An overview of our truth discovery algorithm}
By combining the proposed task assignment and truth inference algorithms, we develop a novel \emph{crowdsourced truth discovery algorithm using hierarchies}.
As illustrated in Figure~\ref{fig:td_in_KF}, our algorithm consists of two components: \emph{hierarchical truth inference} and \emph{task assignment}.
The hierarchical truth inference algorithm finds the correct values from the conflicting values, which are collected from different sources and crowd workers, using hierarchies.
The task assignment algorithm distributes objects to the workers who are likely to increase the accuracy of the truth discovery the most.
The proposed \emph{crowdsourced truth discovery algorithm} repeatedly alternates the truth inference and task assignment until the budget of crowdsourcing runs out.
\red{
	As discussed in \cite{li2016crowdsourced}, some workers answer slower than others and increase the latency.
	However, we do not investigate how to reduce the latency in this work since we can utilize the techniques proposed in \cite{haas2015clamshell}.
}

\begin{figure}[!t]
\centering
\includegraphics[width=3.1in]{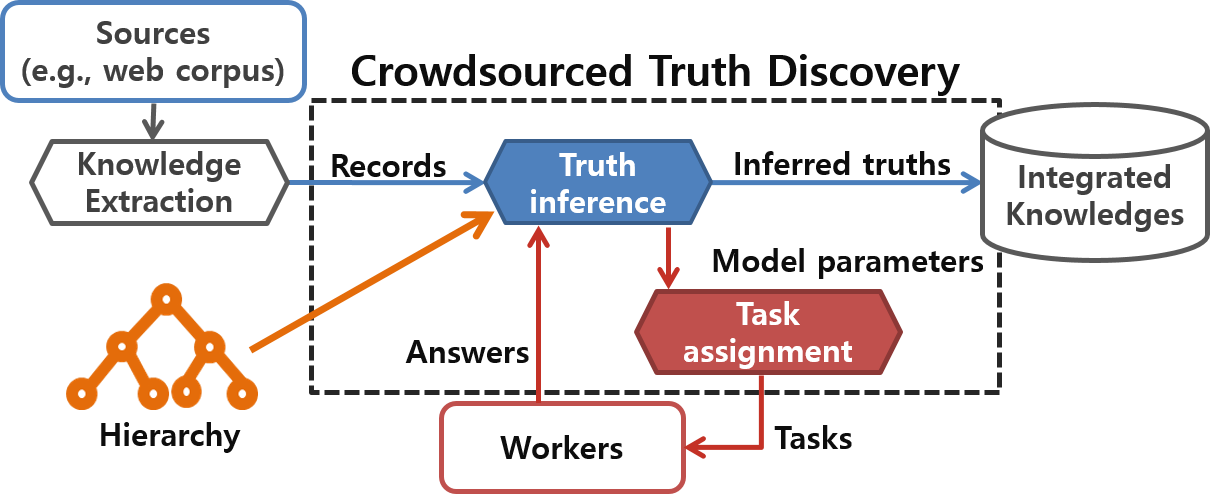}
\vspace{-0.1in}
\caption{Crowdsourced truth discovery in KF}
\vspace{-0.1in}
\label{fig:td_in_KF}
\end{figure}


\iffalse
The contributions of this paper are summarized below.
\begin{itemize}[topsep=3pt,itemsep=1ex,partopsep=1ex,parsep=1ex]
\item We propose a truth inference algorithm which utilizes hierarchical structures in claimed values. 
Note that this is the first work which considers both the reliabilities and the tendencies of generalization of sources.
\item We define a quality measure for a task that will improve the accuracy of truth discovery the most and devise an efficient algorithm to compute the measure.
We also proposed an efficient task assignment algorithm for multiple crowd workers based on the quality measure.
\item We empirically show that the proposed algorithm outperforms the existing works with extensive experiments on real-life datasets.   
\end{itemize}
\else

\minisection{Our contributions}
\eat{We propose a truth inference algorithm which utilizes the hierarchical structures in claimed values. 
To the best of our knowledge, this is the first work which considers both the reliabilities and the generalization tendencies of the sources.
To assign a task which will most improve the accuracy, we develop an incremental EM algorithm to estimate the accuracy improvement by a task.
In addition, we devise an efficient task assignment algorithm for multiple crowd workers based on the quality measure.
We empirically show with extensive experiments on real-life datasets that the proposed algorithm outperforms the existing works. }
\fi
\eat{
We propose a truth inference algorithm utilizing the hierarchical structures in claimed values.
	To the best of our knowledge, it is the first work which considers both the reliabilities and the generalization tendencies of the sources.
 To assign a task which will most improve the accuracy, we develop an incremental EM algorithm to estimate the accuracy improvement for a task by considering the number of claimed values as well as the confidence distribution. 
	We also devise an efficient task assignment algorithm for multiple crowd workers based on the quality measure.
 We empirically show with extensive experiments on real-life datasets that the proposed algorithm outperforms the existing works. }
\red{
The contributions of this paper are summarized below.
\begin{itemize}
	\item We propose a truth inference algorithm utilizing the hierarchical structures in claimed values.
	To the best of our knowledge, it is the first work which considers both the reliabilities and the generalization tendencies of the sources.
	\item To assign a task which will most improve the accuracy, we develop an incremental EM algorithm to estimate the accuracy improvement for a task by considering the number of claimed values as well as the confidence distribution. 
	 We also devise an efficient task assignment algorithm for multiple crowd workers based on the quality measure.
	\item We empirically show that the proposed algorithm outperforms the existing works with extensive experiments on real-life datasets. 
\end{itemize}
}
\eat{
The rest of this paper is organized as follows. 
After discussing related work in Section~\ref{sec:related_work}, we provide the problem definitions in Section~\ref{sec:preliminaries}.
We next propose a truth inference algorithm and a task assignment algorithm in Section~\ref{TruthInference} and \ref{sec:task_assignment}, respectively. 
We provide the performance study in Section~\ref{sec:experiment} and conclude in Section~\ref{sec:conclusion}.}

\section{Preliminaries}
\label{sec:preliminaries}


In this section, we provide the definitions and the problem formulation of \emph{crowdsourced truth discovery in the presence of hierarchy}.
\subsection{Definitions}


For the ease of presentation, 
we assume that we are interested in a single attribute of objects although our algorithms can be easily generalized to find the truths of multiple attributes.
Thus, we use `the target attribute value of an object'
and `the value of an object' interchangeably.

A \emph{source} is a structured or unstructured database which contains the information on target attribute values for a set of objects.
In this paper, a \emph{source} is a certain web page or website and a \emph{worker} represents a human worker in crowdsourcing platforms. 
The information of an object provided by a source or a worker is called a \emph{claimed value}. 
\eat{We define a record and an answer as follows:}
 

\begin{definition}
\vspace{-0.02in}
A \emph{record} is a data describing the information about an object from a source. A record on an object $o$ from a source $s$ is represented as a triple $(o, s, v_o^s)$ where $v_o^s$ is the claimed value of an object $o$ collected from $s$.
Similarly, if a worker $w$ answers that the truth on an object $o$ is $v_o^w$, the \emph{answer} is represented as $(o, w, v_o^w)$. 
\end{definition}

Let $S_o$ be the set of the sources which claimed a value on the object $o$ and $V_o$ be the set of candidate values collected from $S_o$.
Each worker in $W_o$ answers a question about the object $o$ by selecting a value from $V_o$.

In our problem setting, we assume that we have a hierarchy tree $H$ of the claimed values.
If we are interested in an attribute related to locations (e.g., birthplace), $H$ would be a geographical hierarchy with different levels of granularity (e.g., continent, country, city, etc.). 
We also assume that there is no answer with the value of the root in the hierarchy since it provides no information at all (e.g., Earth as a birthplace). 
\ifnotationsummary
We summarize the notations to be used in the paper in Table~\ref{tab:notations}.
\fi

\begin{example}
Consider the records in Table~\ref{expl_records}.
Since the source Wikipedia claims that the location of the Statue of Liberty is Liberty Island, it is represented by $v_o^s=$`Liberty Island' where $o=$`Statue of Liberty' and $s=$`Wikipedia'.
If a human worker `Emma Stone' answered Big Ben is in London, it is represented by $v_o^w=$`London' where $o=$`Big Ben' and $w=$`Emma Stone'. 
\end{example}


\subsection{Problem Definition}
\eat{
We formally define the problem of crowdsourced truth discovery in the presence of hierarchies. 
\begin{definition}
Given a set of objects $O$ and the corresponding hierarchy tree $H$, we define two subproblems of the crowdsourced truth discovery. 

\begin{itemize} 
\item \emph{Hierarchical truth inference problem:}
Given a set of records $R$ collected from the sources and a set of answers $A$ from the workers,
we find the \red{most specific} true value ${v_o^*}$ of each object $o \in O$ among the candidate values in $V_o$ by using the hierarchy $H$.
\item \emph{Task assignment problem:}
For each worker $w$ in a set of workers $W$,
we select the top-$k$ objects from $O$ which are likely to increase the overall accuracy of the inferred truths the most by using the hierarchy $H$.
\end {itemize}
\end{definition}}
Given a set of objects $O$ and a hierarchy tree $H$, we define the two subproblems of the crowdsourced truth discovery. 
\begin{definition}[Hierarchical truth inference problem]
	For a set of records $R$ collected from the sources and a set of answers $A$ from the workers,
	we find the \red{most specific} true value ${v_o^*}$ of each object $o \in O$ among the candidate values in $V_o$ by using the hierarchy $H$.	
\end{definition}
\begin{definition}[Task assignment problem]
	For each worker $w$ in a set of workers $W$,
	we select the top-$k$ objects from $O$ which are likely to increase the overall accuracy of the inferred truths the most by using the hierarchy $H$.
\end{definition}
We present a hierarchical truth inference algorithm in Section~\ref{TruthInference} and a task assignment algorithm in Section~\ref{sec:task_assignment}.

\ifnotationsummary
\begin{table}[b]
	\center
	\vspace{-0.05in}
	\caption{Notations}
	\vspace{-0.03in}
	\label{tab:notations}
	\small
	\begin{tabular}{c|l}
		\toprule
		Symbol& Description\\
		\midrule
		\midrule
		$s$ & A data source \\
		$w$ & A crowd worker \\
		$v_o^s$ & Claimed value from $s$ about $o$ \\
		$v_o^w$ & Claimed value from $w$ about $o$ \\
		\midrule
		$R$ & \makecell[l]{Set of all records collected from the set of sources $S$}\\
		$A$ & \makecell[l]{Set of all answers collected from the set of workers $W$}\\
		$V_o$ & Set of candidate values about $o$\\
		\midrule
		$S_o$ & Set of sources which post information about $o$\\
		$W_o$ & Set of workers who answered about $o$\\
		$O_s$ & Set of objects that source $s$ provided a value\\
		$O_w$ & Set of objects that worker $w$ answered to\\
		$G_o(v)$ & \makecell[l]{Set of values in $V_o$ which are ancestors of a value $v$ \\ ~except the root in the hierarchy $H$}\\
		$D_o(v)$ & Set of values in $V_o$ which are descendants of $v$ \\			
		\bottomrule
	\end{tabular}
\end{table}
\fi

\section{Hierarchical Truth Inference}
\label{TruthInference}
For the hierarchical truth inference, we first model the trustworthiness of sources and workers for a given hierarchy.
Then, we propose a probabilistic model to describe the process of generating the set of records and the set of answers based on the trustworthiness modeling.
We next develop an inference algorithm to estimate the model parameters and determine the truths.

\subsection{Our Generative Model}
\label{sec:gen_model}

\begin{figure}[tb]
\centering
\includegraphics[width=2.2in]{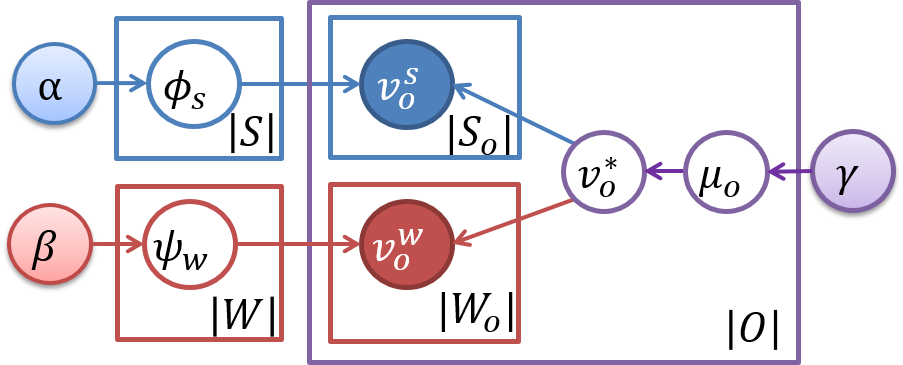}
\vspace{-0.1in}
\caption{A graphical model for truth inference}
\vspace{-0.1in}
\label{fig:graphical_model}
\end{figure}

Our probabilistic graphical model in Figure~\ref{fig:graphical_model} expresses the conditional dependence (represented by edges) between random variables (represented by nodes).
While the previous works \cite{demartini2012zencrowd,whitehill2009whose,karger2011iterative,raykar2010learning}  assume that all sources and workers have their own reliabilities only,
we assume that each source or worker has its generalization tendency as well as reliability.
We first describe how sources and workers generate the claimed values based on their trustworthiness. 
We next present the model for generating the true value.
Finally, we provide the detailed generative process of our probabilistic model. 

\minisection{Model for source trustworthiness}
For an object $o$, let $v_o^*$ be the truth and $v_o^s$ be the claimed value reported by a source $s$.
Recall that $V_o$ is the set of candidate values for an object $o$. 
Furthermore, we let $G_o(v)$ denote the set of candidate values which are ancestors of a value $v$ except for the root in the hierarchy $H$.

There are three relationships between a claimed value $v_o^s$ and the truth $v_o^*$: (1) $v_o^s = v_o^*$, (2) $v_o^s \in G_{o}(v_o^*)$ and (3) otherwise.
Let $\phi_s=(\phi_{s,1},\phi_{s,2},\phi_{s,3})$ be the \emph{trustworthiness distribution} of a source $s$ 
where $\phi_{s,i}$ is the probability that a claimed value of the source $s$ corresponds to the $i$-th relationship.
In each relationship, a claimed value is generated as follows:
\begin{itemize} 
\item {\bf Case 1 ($v_o^s = v_o^*$):} The source $s$ provides the exact true value with a probability $\phi_{s,1}$.

\item {\bf Case 2 ($v_o^s \in G_o(v_o^*)$):} 
The source $s$ provides a \emph{generalized true value} $v_o^s$ with a probability $\phi_{s, 2}$. 
In this case, the claimed value is an ancestor of the truth $v_o^*$ in $H$.
We assume that the claimed value is uniformly selected from $G_{o}(v_o^*)$.

\item {\bf Case 3 ($otherwise$):} 
The source $s$ provides a wrong value $v_o^s$ not even in $G_o(v_o^*)$.
The claimed value is uniformly selected among the rest of the candidate values in $V_o$.

\end{itemize}

The probability distribution $\phi_s$ is an initially-unknown model parameter to be estimated in our inference algorithm.
Accordingly, the probability of selecting an answer $v_o^s$ among the values in $V_o$ for an object $o$ is represented by
\ifshrinkmode
{\small
\begin{equation}
\label{eq:sourceProb}
P(v_o^s|v_o^*,\phi_s) = \begin{cases}
\phi_{s,1} & \mbox{if $v_o^s=v_o^*$,}\\
\phi_{s,2}/|G_o(v_o^*)| & \mbox{if $v_o^s \in G_o(v_o^*)$,} \\
\phi_{s,3}/(|V_o|-|G_o(v_o^*)|-1) & \mbox{otherwise.}
\end{cases}
\vspace{-0.02in}
\end{equation}
}
\fi

\noindent For the prior of the distribution $\phi_{s}$, 
we assume that it follows a Dirichlet distribution $Dir(\alpha)$, with a hyperparameter $\alpha=(\alpha_1,\alpha_2,$ $\alpha_3)$, 
which is the conjugate prior of categorical distributions.

Let $O_H$ be the set of objects who have an ancestor-descendant relationship in their candidate set.
In practice, there may exist some objects whose candidate values do not have an ancestor-descendant relationship.
In this case, the probability of the second case (i.e., $\phi_{s,2}$) may be underestimated.
Thus, if there is no ancestor-descendant relationship between the claimed values about $o$ (i.e., $o\notin O_H$), 
we assume that a source generates its claimed value $v_o^s$ with the following probability
\ifshrinkmode
{\small
\begin{equation}
\label{eq:sourceProb_noH}
P(v_o^s|v_o^*,\phi_s)= \begin{cases}
\phi_{s,1}+\phi_{s,2} & \mbox{if $v_o^s=v_o^*$,}\\
\phi_{s,3}/(|V_o|-1) & \mbox{otherwise.}
\end{cases}
\vspace{-0.03in}
\end{equation}
}
\fi

\begin{figure*}
\hspace{-0.55in}
\subfloat{\parbox{0.01\textwidth}{
\begin{eqnarray}
f_{o,s}^v \eqmid  \frac{P(v_o^s|v_o^*\eqmid v,\phi_s)\cdotmid \mu_{o,v}}{\sum_{v' \inmid V_o}{P(v_o^s|v_o^*\eqmid v',\phi_s) \cdotmid \mu_{o,v'}}} \nonumber\\
	f_{o,w}^v \eqmid \frac{P(v_o^w|v_o^*\eqmid v,\psi_w)\cdotmid \mu_{o,v}}{\sum_{v'\inmid V_o}{P(v_o^w|v_o^*\eqmid v',\psi_w) \cdotmid \mu_{o,v'}}}\nonumber
\end{eqnarray}}
}
\hspace{-1.31in}
\subfloat{\parbox{0.01\textwidth}{
\begin{align}
	g_{o,s}^1
	\eqmid & \frac{\phi_{s,1}\cdotmid \mu_{o,v_o^s}}{\sum_{v \inmid V_o}{P(v_o^s|v_o^*\eqmid v,\phi_s)\cdotmid \mu_{o,v}}}\nonumber
	\\
	g_{o,s}^2\eqmid &
	\begin{cases}
	 \frac{\sum_{v \in D_o(v_o^s)}{\frac{\phi_{s,2}}{|G_o(v)|}\cdot\mu_{o,v}}}{\sum_{v \in V_o}{P(v_o^s|v_o^*\eqmid v,\phi_s)\cdot\mu_{o,v}}} & \mbox{\scriptsize if  $o \in O_H$}\\[5pt] 
	 \frac{\phi_{s,2}\cdot\mu_{o,v_o^s}}{\sum_{v \inmid V_o}{P(v_o^s|v_o^*\eqmid v,\phi_s) \cdot \mu_{o,v}}} & \mbox{\scriptsize otherwise}
	 \end{cases}	
	\nonumber
	\\
	g_{o,s}^3
	\eqmid & \frac{\sum_{v \in \neg D_o(v_o^s)}{\frac{\phi_{s,3}}{|V_o - G_o(v)|-1}\cdot\mu_{o,v}}}{\sum_{v \in V_o}{P(v_o^s|v_o^*\eqmid v,\phi_s)\cdot\mu_{o,v}}}\nonumber
\end{align}}
}
\hspace{-0.31in}
\subfloat{\parbox{0.01\textwidth}{
\begin{align}
	g_{o,w}^1
	\eqmid & 
	\frac{\psi_{w,1}\cdot\mu_{o,v_o^w}}{\sum_{v \in V_o}{P(v_o^w|v_o^*\eqmid v,\psi_w)\cdot\mu_{o,v}}} 		\nonumber 
	\\
	g_{o,w}^2
	\eqmid &
	\begin{cases}	
	\frac{\sum_{v \inmid D_o(v_o^w)}{\psi_{w,2}\cdot Pop_{2}(v_o^w|v_o^* \eqmid v) \cdot\mu_{o,v}}}{\sum_{v \in V_o}{P(v_o^w|v_o^*\eqmid v,\psi_w)\cdot\mu_{o,v}}}& \mbox{\scriptsize if  $o \in O_H$}\\[5pt] 
	\frac{\psi_{w,2}\cdot\mu_{o,v_o^w}}{\sum_{v \in V_o}{P(v_o^w|v_o^*\eqmid v,\psi_w)\cdot\mu_{o,v}}} & \mbox{\scriptsize otherwise}
	\end{cases}	
	\nonumber
	\\
	g_{o,w}^3
	\eqmid & \frac{\sum_{v \in \neg D_o(v_o^w)}{\psi_{w,3}\cdot Pop_{3}(v_o^w|v_o^* \eqmid v) \cdot\mu_{o,v}}}{\sum_{v \in V_o}{P(v_o^w|v_o^*\eqmid v,\psi_w)\cdot\mu_{o,v}}}\nonumber
\end{align}}
}
\hspace{-0.05in}
\vspace{-0.08in}
\caption{E-step for the proposed truth inference algorithm}
\label{fig:Estep}
\vspace{-0.06in}
\end{figure*}

\minisection{Model for worker trustworthiness}
Let $v_o^w$ be the claimed value chosen by a worker $w$ among the candidates in $V_o$ for an object $o$.
Similar to the model for source trustworthiness, we also assume the three relationships between a claimed value $v_o^w$ and the truth $v_o^*$: (1) $v_o^w = v_o^*$, (2) $v_o^w \in G_{o}(v_o^*)$ and (3) otherwise. 
Each worker $w$ has its \emph{trustworthiness distribution} $\psi_w\eqmid(\psi_{w,1},\psi_{w,2},\psi_{w,3})$ where $\psi_{w,i}$ is the probability that an answer of the worker $w$ corresponds to the $i$-th relationship.
We assume that the trustworthiness distribution is generated from $Dir(\beta)$ with a hyperparameter $\beta=(\beta_1,\beta_2,\beta_3)$.
\eat{Note that a worker may answer based on the claimed values from the sources.
For example, s}

Since it is difficult for the workers to be aware of the correct answer for every object, a worker can refer to web sites to answer the question.
In such a case, if there is a widespread misinformation across multiple sources, the worker is also likely to respond with the incorrect information.
Similar to \cite{dong2012less,LCA}, we thus exploit the \emph{popularity} of a value in Cases 2 and 3 to consider such dependency between sources and workers.

\begin{itemize} 
\item {\bf Case 1 ($v_o^w = v_o^*$):} The worker $w$ provides the exact true value with a probability $\psi_{w,1}$.

\item {\bf Case 2 ($v_o^w \in G_o(v_o^*)$):} 
The worker $w$ provides a generalized true value with a probability $\psi_{w, 2}$. 
We assume that the claimed value $v_o^w$ is selected according to the popularity $Pop_{2}(v_o^w|v_o^*)= \frac{|\{s | s\in S_o, v_o^s = v \}|}{|\{s | s\in S_o, v_o^s \in G_o(v_o^*) \}|}$ which is the proportion of the records whose claimed value is $v_o^w$ out of the records with generalized values of $v_o^*$.
\item {\bf Case 3 ($otherwise$):} 
The claimed value is selected from the wrong values according to the popularity $Pop_{3}(v_o^w|v_o^*)= \frac{|\{s | s\in S_o, v_o^s = v \}|}{|\{s | s\in S_o, v_o^s \notin G_o(v_o^*), v_o^s \neq v_o^* \}|}$. 
\end{itemize}

%



By the above model, the probability of selecting an answer $v_o^w$ for the truth $v_o^*$ of an object $o$ is formulated as
{\small
\begin{equation}
\label{eq:workerProb}
P(v_o^w|v_o^*,\psi_w) \eqmid \begin{cases}
        \psi_{w,1} & \mbox{if $v_o^w=v_o^*$,}\\
        \psi_{w,2}\cdot Pop_{2}(v_o^w|v_o^*) & \mbox{if $v_o^w \in G_o(v_o^*)$,} \\
        \psi_{w,3}\cdot Pop_{3}(v_o^w|v_o^*) & \mbox{otherwise.}
        \end{cases}
\end{equation}
}

\noindent Similar to the model for source trustworthiness, if there is no ancestor-descendant relationship in the candidate values of an object $o$, the probability of selecting a claimed value $v_o^w$ is
{\small
\begin{equation}
\label{eq:workerProb_noH}
P(v_o^w|v_o^*,\psi_w)= \begin{cases}
        \psi_{w,1}+\psi_{w,2} & \mbox{if $v_o^w=v_o^*$,}\\
        \psi_{w,3}\cdot Pop_{3}(v_o^w|v_o^*) & \mbox{otherwise.}
        \end{cases}
\end{equation}}

\minisection{Model for truth}
\eat{In our unsupervised problem setting, since the correct value is not known, w}
We introduce the probability distribution over the candidate answers to determine the truth, called \emph{confidence distribution}.
Each object $o$ has a confidence distribution $\mu_o=\{\mu_{o,v}\}_{v\in V_o}$
where $\mu_{o,v}$ is the probability that the value $v \in V_o$  is the true answer for $o$.
We also use a dirichlet prior $Dir(\gamma_o)$ for the confidence distribution $\mu_o$ where $\gamma_o=\{\gamma_{o,v}\}_{v\in V_o}$ is a hyperparameter.

Based on the above three models, the generative process of our model works as follows.


\minisection{Generative process}
Given a set of objects $O$, a set of sources $S$ and a set of workers $W$,
our proposed model assumes the following generative process for the set of records $R$ and the set of answers $A$:

\begin{enumerate} [topsep=1pt,itemsep=0ex,partopsep=0.4ex,parsep=0.4ex]
\item Draw $\phi_s \sim Dir(\alpha)$ for each source $s\in S$
\item Draw $\psi_w \sim Dir(\beta)$ for each worker $w\in W$
\item For each object $o\in O$
	\begin{enumerate} [topsep=1pt,itemsep=0ex,partopsep=0.4ex,parsep=0.4ex]
	\item Draw $\mu_o \sim Dir(\gamma_o)$
	\item Draw a true value $v_o^* \sim Categorical(\mu_o)$
	\item For each source $s\in S_o$
		\begin{enumerate} [topsep=1pt,itemsep=0ex,partopsep=0.4ex,parsep=0.4ex]
		\item Draw a value $v_o^s$ following $P(v_o^s|v_o^*,\phi_s)$\label{l123}
		\end {enumerate}
	\item For each worker $w \in W_o$
		\begin{enumerate} [topsep=1pt,itemsep=0ex,partopsep=0.4ex,parsep=0.4ex]
		\item Draw a value $v_o^w$ following $P(v_o^w|v_o^*,\psi_w)$
		\end {enumerate}	
	\end{enumerate}
\end{enumerate}

\subsection{Estimation of Model Parameters}
\label{sec:emalgorithm}
We now develop an inference algorithm for the generative model.
Let $\Theta=\pmb{\phi} \cup \pmb{\psi} \cup \pmb{\mu}$ be the set of all model parameters where $\pmb{\phi}\eqmid\{\phi_s|s\inmid S\}$, $\pmb{\psi}\eqmid\{\psi_w|w\inmid W\}$ and $\pmb{\mu}\eqmid\{\mu_o|o\inmid O\}$.
We propose an EM algorithm to find the maximum a posteriori (MAP) estimate of the parameters in our model.

\minisection{The maximum a posteriori (MAP) estimator}
Recall that $R = \{(o,s,v_o^s)\}$ is the set of records from the sources and $A\eqmid \{(o,w,v_o^w)\}$ is the set of answers from the workers. 
For every object $o$, each source $s \in S_o$ and each worker $w \in W_o$ generates its claimed values independently.  
Then, the likelihood of $R$ and $A$ based on our generative model is 
{\small
\begin{equation*}
\vspace{-0.02in}
P(R,A|\Theta) \eqmid \prod_{o\in O}{\prod_{s\in S_o} P(v_o^s | \phi_s, \mu_o)} \cdotmid \prod_{o \in O}{\prod_{w \in W_o} P(v_o^w | \psi_w,\mu_o)}
\vspace{-0.02in}
\end{equation*}}

\noindent where the probability of generating a claimed value by a source or a worker becomes
{\small
\begin{align}
	P(v_o^s | \phi_s,\mu_o) &= \sum_{v \in V_o}{P(v_o^s | \phi_s,v_o^*=v) \cdot \mu_{o,v}}\label{eq:margin_s}\\
	P(v_o^w | \psi_w,\mu_o) &= \sum_{v \in V_o}{P(v_o^w | \psi_w,v_o^*=v) \cdot \mu_{o,v}}.\label{eq:margin_w}
\end{align}}


\noindent Consequently, the MAP point estimator is obtained by maximizing the log-posterior as
\begin{equation}
\label{eq:map}
\hat{\Theta} = \argmax_{\Theta}{\{\log{P(R,A|\Theta)}+\log{P(\Theta)}\}} = \argmax_{\Theta}{\mathbb{F}}  
\end{equation}
where the objective function $\mathbb{F}$ is 
{\small
\begin{align}
\vspace{-0.02in}
\mathbb{F}
& =\sum_{o\in O}{\sum_{s\in S_o} \log{ \sum_{v \in V_o}{P(v_o^s | \phi_s,v_o^*=v) \cdot \mu_{o,v}}}} \nonumber \\
& + \sum_{o\in O}{\sum_{w\in W_o}\log{\sum_{v \in V_o}{P(v_o^w | \psi_w,v_o^*=v) \cdot \mu_{o,v}}}} \label{eq:objectiveL}\\
& + \sum_{s\in S}{\log{p(\phi_s|\alpha)}} \plusmid 
\sum_{w\in W}{\log{p(\psi_w|\beta)}} \plusmid \sum_{o\in O}{\log{p(\mu_o|\gamma_o)}}.\nonumber
\vspace{-0.02in}
\end{align}}

Note that although we assumed that each claimed value is generated independently according to its probability distribution defined in Eq.~(\ref{eq:margin_s}) and (\ref{eq:margin_w}), the dependencies between sources and workers are already considered in $Pop_{2}(v_o^w|v_o^*)$ and $Pop_{3}(v_o^w|v_o^*)$.

\minisection{The EM algorithm}
We introduce a random variable $C_{v}$ to represent the type of the relationship between the claimed value $v$ and the truth $v_o^*$.
It is defined as follows:
\begin{equation*}
\vspace{-0.03in}
C_{v} = \begin{cases}
        1 & \mbox{if $v=v_o^*$,}\\
        2 & \mbox{if $v \in G_o(v_o^*)$,} \\
        3 & \mbox{otherwise.}
        \end{cases}
\vspace{-0.02in}
\end{equation*}


In the \emph{\bf E-step}, we compute the conditional distributions of the hidden variables $C_{v_o^s}$, $C_{v_o^w}$ and $v_o^*$ under our current estimate of the parameters $\Theta$.
Let $f_{o,s}^v$, $f_{o,w}^v$, $g_{o,s}^t$ and $g_{o,w}^t$ denote the conditional probabilities
$P(v_o^*\eqmid v| v_o^s, \mu_o, \phi_s)$, $P(v_o^*\eqmid v | v_o^w, \mu_o, \psi_w)$, $P(C_{v_o^s}\eqmid t | \mu_o, \phi_s)$ and $P(C_{v_o^w}\eqmid t|$  $\mu_o, \psi_w)$, respectively.
Using Bayes' rule, we can update the conditional probabilities as shown in Figure~\ref{fig:Estep} where $D_o(v) = \{v'|v \in G_o(v') \wedge v'\in V_o \}$ is the set of descendants of $v$ among the candidate values and $\neg D_o(v) = V_o \minusmid D_o(v) \minusmid \{v\}$ is the set of candidate values each of which is neither a descendant of the value $v$ nor the $v$ itself.

In the \emph{\bf M-step}, we find the model parameters $\Theta$ that maximize our objective function $\mathbb{F}$.
We first add Lagrange multipliers to enforce the constraints of model parameters.

{\scriptsize
\vspace{-0.05in}
\begin{equation*}
\mathbb{L} = \mathbb{F}
\plusmid \sum_{s \in S}{\lambda_{\phi,s}  \left( 1\minusmid\sum_{t=1}^{3}{\phi_{s,t}} \right)}
\plusmid \sum_{w \in W}{\lambda_{\psi,w}  \left( 1 \minusmid \sum_{t=1}^{3}{\psi_{w,t}} \right)} 
\plusmid \sum_{o \in O}{\lambda_{\mu,o}  \left( 1 \minusmid \sum_{v \in V_o}{\mu_{o,v}} \right)}
\end{equation*}
\vspace{-0.04in}
}

We obtain the following equations for updating the model parameters $\Theta$
by taking the partial derivative of the Lagrangian $\mathbb{L}$ with respect to each model parameter and setting it to zero:

{\small
\vspace{-0.03in}
\begin{equation}
\label{eq:MstepMu}
\mu_{o,v} = \frac{\sum_{s\in S_o}{f_{o,s}^v}+\sum_{w\in W_o}{f_{o,w}^v}+\gamma_{o,v} -1}
{|S_o|+|W_o|+\sum_{v'\in V_o}{\left( \gamma_{o,v'}-1 \right)}}
\end{equation}
\begin{equation}
\label{eq:MstepPhi}
\phi_{s,t} = \frac{\sum_{o\in O_s}{g_{o,s}^t}+\alpha_t-1}{|O_s|+\sum_{t'=1}^{3}{(\alpha_{t'}-1)}}
\end{equation}
\begin{equation}
\label{eq:MstepPsi}
\psi_{w,t} = \frac{\sum_{o\in O_w}{g_{o,w}^t}+\beta_t-1}{|O_w|+\sum_{t'=1}^{3}{(\beta_{t'}-1)}}
\end{equation}}
\hspace{-0.03in}where $O_s$ and $O_w$ are the sets of objects claimed by $s$ and $w$, respectively.
We infer the truth by choosing the value with the maximum confidence among the candidate values as 
\begin{equation}
\label{eq:estimated_truth}
v_o^* = \argmax_{v \in V_o} \mu_{o,v}.
\end{equation}

\minisection{Extension to numerical data}
In the world wide web, numerical data also have an implicit hierarchy due to the significant digits which carry meaning contributing to its measurement resolution.
For example, even though the area of Seoul is $605.196 km^2$,  
different websites may represent the area in various forms depending on the significant figures (e.g., $605.2 km^2$, $605 km^2$).
\eat{As we discussed earlier, the existing methods for categorical data regards the values are completely different.
On the other hand, a}An existing algorithm \cite{li2014confidence} to handle numerical data utilizes a weighted sum of the claimed values to consider the distribution of the claimed values.
However, such method is sensitive to outliers and thus need a proper preprocessing to remove the outliers.
To overcome the drawbacks, we generate the underlying hierarchy in the numerical data by assuming that $v_d$ is a descendant of $v_a$ if a value $v_a$ can be obtained by rounding off a value $v_d$.
Then, we can use our TDH algorithm to find the truths in numerical data by taking  into account the relationship between the values in the implicit hierarchy.
Our algorithm is also robust to the outliers with extremely small or large value since we estimate the truth by selecting the most probable value from the candidate values rather than computing a weighted average of the claimed values.

\section{Task Assignment to Workers}
\label{sec:task_assignment}
In this section, 
we propose a task assignment method to select 
the best objects to be assigned to the workers in crowdsourcing systems.
We first define a quality measure of tasks called \emph{Expected Accuracy Increase (EAI)}
 and develop an incremental EM algorithm to quickly estimate the quality measure.
Finally, we present an efficient algorithm for assigning the $k$ questions to each worker $w$ in a set of workers $W$ based on the measure.

\subsection{The Quality Measure}
\label{sec:quality_measure}
Given a worker $w$, our goal is to choose an object to be assigned to the worker $w$ which is likely to increase the accuracy of the estimated truths the most.
Thus, we define a quality measure for a pair of worker and an object based on the improvement of the accuracy.
As discussed in \cite{zheng2015qasca}, the improvement of the accuracy by a task can be estimated by using the difference between the highest confidence as follows:
{\small
\begin{equation}
\label{eq:accimprovement}
\ifshrinkmode
(Accuracy~improvement) = \{\max_{v}{\mu_{o,v|w}}-\max_{v}{\mu_{o,v}}\}/|O|
\fi
\end{equation}}
\hspace{-0.07in} where $\mu_{o,v|w}$ is the estimated confidence on $v$ if the worker $w$ answers about an object $o$.

\minisection{The quality measure used by QASCA}
The QASCA\cite{zheng2015qasca} algorithm calculates the estimated confidence by using the current confidence distribution and the likelihood of the answer $v_o^w$ given the truth $v_o^*=v$ as
{\small
\begin{equation*}
\mu_{o,v|w} \propto \mu_{o,v} \cdot p(v_o^w = v'|v_o^*=v)
\end{equation*}}

\noindent where $v'$ is a sampled claimed value.
There are two drawbacks in the quality measure of QASCA. 
First, since it computes the estimated confidence $\mu_{o,v|w}$ based on a sampled answer $v_o^w=v$,
the value of the quality measure is very sensitive to the sampled answer.
In addition, QASCA does not consider the number of claimed values collected so far and the estimated confidence $\mu_{o,v|w}$ may not be accurate.
For instance, assume that there exist two objects which have identical confidence distributions.
If one of the objects already has many collected claimed values, an additional answer is not likely to change the confidence significantly.
Thus, task assignment algorithms should select another object who has a smaller number of collected records and answers.

\minisection{Our quality measure}
To avoid the sensitiveness caused by sampling answers, we develop a new quality measure \emph{Expected Accuracy Improvement (EAI)} which is obtained by taking the expectation to Eq.~(\ref{eq:accimprovement}). That is,
{\small
\begin{equation}
\label{eq:def_emci}
\ifshrinkmode
EAI(w,o) = \{E[\max_{v} {\mu_{o,v|w}}] - \max_{v} {\mu_{o,v}}\}/|O|.
\fi
\end{equation}}
\noindent By the definition of expectation, $E[\max_{v} {\mu_{o,v|w}}]$ becomes
{\small
\begin{equation}
\label{eq:expected_max_mu}
\begin{split}
&E[\max_{v} {\mu_{o,v|w}}] \eqmid \sum_{v'\in V_o}{P(v_o^w \eqmid v'| \psi_w, \mu_o) \cdotmid \max_{v} {\mu_{o,v|v_o^w \eqmid v'}}}.
\end{split}
\end{equation}}
\hspace{-0.06in} where $\mu_{o,v|v_o^w = v'}$ is the conditional confidence when a worker $w$ answers with $v'$ about the object $o$.

Since $P(v_o^w=v'| \psi_w, \mu_o)$ can be computed by Eq.~(\ref{eq:margin_w}),
to compute $E[\max_{v} {\mu_{o,v|w}}]$ by Eq.~(\ref{eq:expected_max_mu}), we need the estimation of the conditional confidence $\mu_{o,v|v_o^w=v'}$ with an additional answer $v_o^w=v'$.
Recall that the estimated confidence computed by QASCA may not be accurate because it does not consider the collected records and answers so far.
To reduce the error, we use them to compute the conditional confidence  $\mu_{o,v|v_o^w=v'}$.
We can compute the conditional confidence $\mu_{o,v|v_o^w=v'}$ by applying the EM algorithm in Section~\ref{sec:emalgorithm} with the collected records and answers including $v_o^w=v'$.
However, since it is computationally expensive, we next develop an \emph{incremental EM algorithm}. 

\subsection{The Incremental EM Algorithm}

Let $\mathbb{F}_{v_o^w=v'}$ be the objective function in Eq.~(\ref{eq:map}) after obtaining an additional answer $(o,w,v')$.
Then, we have
\begin{equation}
\mathbb{F}_{v_o^w=v'}=\mathbb{F}+\log{\sum_{v \in V_o}{P(v_o^w=v' | \psi_w,v_o^*=v) \cdot \mu_{o,v}}}\nonumber
\end{equation}
by adding the related term of the additional answer (log likelihood of the additional answer) to Eq.~(\ref{eq:objectiveL}).
Instead of running the iterative EM algorithm in Section~\ref{sec:emalgorithm}, we incrementally perform a \emph{single EM-step} to speed up for only the additional answer with the current model parameters and the above objective function.

\minisection{E-step}
Since we use the current model parameters, the probabilities of the hidden variables for collected records and answers are not changed.
Thus, we only need to compute the conditional probabilities of the hidden variable given the additional answer as
\begin{equation}
\label{eq:inc_estep}
f_{o,w|v_o^w=v'}^{v} \eqshort \frac{P(v_o^w\eqmid v'|v_o^*\!\eqshort v,\psi_w)\cdot\mu_{o,v}}{\sum_{v''\in V_o}{P(v_o^w \eqmid v'|v_o^* \eqshort v'',\psi_w) \cdot \mu_{o,v''}}}
\end{equation}
based on the equation for $f_{o,w}^v$ used at the E-step in Figure~\ref{fig:Estep}.

\minisection{M-step}
For the objective function $\mathbb{F}_{v_o^w=v'}$, we obtain the following equation of the M-step for the confidence distribution $\mu_o$ with the additional answer $v_o^w=v'$
\begin{equation*}
\label{eq:MstepMu_cond}
\mu_{o,v|v_o^w=v'} \eqshort \frac{\sum_{s\in S_o}{f_{o,s}^v}\!+\!\sum_{w'\in W_o }{f_{o,w'}^v}\!+\!f_{o,w|v_o^w \eqmid v'}^{v}\!+\!\gamma_{o,v} \minusmid 1}
{|S_o|\!+\!|W_o|\!+\!1\!+\!\sum_{v''\in V_o}{\left( \gamma_{o,v''}-1 \right)}}
\end{equation*}
by adding the related terms $f_{o,w|v_o^w=v'}^v$ and $1$ to the numerator and the denominator of the update equation in Eq.~(\ref{eq:MstepMu}), respectively. 
Let $N_{o,v}$ and ${D_o}$ be the numerator and the denominator in Eq.~(\ref{eq:MstepMu}), respectively. 
Then, the above equation can be rewritten as 
\begin{equation}
\vspace{-0.02in}
\label{eq:MstepMu_cond_short}
\mu_{o,v|v_o^w = v'}= \frac{N_{o,v} + f_{o,w|v_o^w=v'}^{v}}{D_o+1}.
\end{equation}

\noindent By substituting $f_{o,w|v_o^w=v'}^v$ in Eq.~(\ref{eq:MstepMu_cond_short}) with Eq.~(\ref{eq:inc_estep}), the conditional confidence becomes
\begin{equation}
\label{eq:conditional_mu}
\mu_{o,v|v_o^w = v'}=\frac{N_{o,v} + \frac{P(v_o^w= v'|v_o^*=v,\psi_w)\cdot\mu_{o,v}}{\sum_{v''\in V_o}{P(v_o^w= v'|v_o^*=v'',\psi_w) \cdot \mu_{o,v''}}}}{D_o+1}.
\end{equation}
Since $N_{o,v}$ and $D_o$ are proportional to the number of the existing claimed values, the confidence will be changed very little if there are many claimed values already. 
Thus, we can overcome the second drawback of QASCA.
Since $N_{o,v}$s and ${D_o}$s are repeatedly used to compute $\mu_{o,v|v_o^w = v'}$, our truth inference algorithm keeps $N_{o,v}$s and ${D_o}$s in main memory to reduce the computation time.

\minisection{Time complexity analysis}
To calculate
 $E[\max_{v} {\mu_{o,v|w}}]$ by Eq.~(\ref{eq:expected_max_mu}), $P(v_o^w=v'| \psi_w, \mu_o)$ is computed $|V_o|$ times and $\mu_{o,v|v_o^w = v'}$ is calculated for every pair of $v$ and $v'$ (i.e., $O(|V_o|^2)$ times).
Moreover, computing $P(v_o^w=v'| \psi_w, \mu_o)$ and $\mu_{o,v|v_o^w = v'}$ take $O(|V_o|)$ time.
Thus, it takes $O(|V_o|^3)$ time to compute $EAI(w,o)$ by Eq.~(\ref{eq:def_emci}).
In reality, $|V_o|$ is very small compared to $|O|$,$|S|$ and $|W|$.
In addition, by utilizing the pruning technique in the next section, we can significantly reduce the computation time.
Therefore, the task assignment step can be performed within a short time compared to the truth inference.
The execution time for each step will be presented in the experiment section.

\subsection{The Task Assignment Algorithm}
\label{sec:upperbound}
To find the $k$ objects to be assigned to each worker, we need to compute $EAI(w,o)$ for all pairs of $w$ and $o$.
To reduce the number of computing $EAI(w,o)$, we develop a pruning technique by utilizing an upper bound of $EAI(w,o)$.



\minisection{An upper bound of EAI}
We provide the following lemma which allows us to compute an upper bound $U_{EAI}(o)$.

\begin{lemma}
\label{lem:UBEAI}
\emph{(Upper Bound of Expected Accuracy Increase)}
For an object $o$ and a worker $w$, 
we have 
\begin{equation}
EAI(w,o) \leq U_{EAI}(o) = \frac{1 - \max_{v} {\mu_{o,v}}}{|O|\cdot(D_o+1)}.
\end{equation} 

\end{lemma}
\begin{proof}
From Eq.~(\ref{eq:conditional_mu}), since $\sum_{v'}{P(v_o^w \eqmid v'| \psi_w, \mu_o)} \eqmid 1$, we get
{\small
\begin{align}
E[\max_{v} {\mu_{o,v|w}}]&=\sum_{v'\in V_o}{P(v_o^w{\eqmid}v'| \psi_w, \mu_o) \cdot \max_{v} {\mu_{o,v|v_o^w = v'}}}  \nonumber\\
&\leq\max_{v,v'} {\mu_{o,v|v_o^w = v'}\cdot\sum_{v'\in V_o}{P(v_o^w{\eqmid}v'| \psi_w, \mu_o) }}   \nonumber\\
&= \max_{v,v'}{\mu_{o,v|v_o^w=v'}}. \mspace{10mu} \mbox{}\label{eq:exp_max_mu_top}
\end{align}
}
Moreover, from Eq.~(\ref{eq:MstepMu_cond_short}), we obtain 
{\small
\begin{align}
\label{eq:cond_mu_ineq}
\mu_{o,v|v_o^w=v'} &= \frac{N_{o,v}+f_{o,w|v_o^w=v'}^{v}}{D_o+1} \leq \frac{N_{o,v}+1}{D_o+1} .
\end{align}}
By substituting Eq.~(\ref{eq:cond_mu_ineq}) for $\mu_{o,v|v_o^w=v'}$ in Eq.~(\ref{eq:exp_max_mu_top}), we derive 
{\small
\begin{equation}
\label{eq:ub_proof_mu}
\begin{split}
    E[\max_{v} {\mu_{o,v|w}}] &\leq \max_{v,v'}{\mu_{o,v|v_o^w \eqmid v'}}
	\leq \frac{\max_{v}N_{o,v}+1}{D_o+1}. 
\end{split}
\end{equation}}
In addition, by applying Eq.~(\ref{eq:ub_proof_mu}) to Eq.~(\ref{eq:def_emci}), we get
{\small
\begin{equation}
\begin{array}{lr}
EAI(w,o) \leq (\frac{\max_{v}{N_{o,v}}+1}{D_o+1} - \max_{v} {\mu_{o,v}})/|O|.
\end{array}
\nonumber
\end{equation}}
Since $\mu_{o,v} = \frac{N_{o,v}}{D_o}$, we finally obtain the upper bound of $EAI(w,o)$.
{\small
\begin{equation}
\begin{array}{ll}
EAI(w,o) &
\leq (\frac{\max_{v}{N_{o,v}}+1}{D_o+1} - \frac{\max_{v}{N_{o,v}}}{D_o})/|O| \\
&= \frac{1 -  \frac{\max_{v}{N_{o,v}}}{D_o}}{|O|\cdot(D_o+1)}
= \frac{1 - \max_{v}{\mu_{o,v}}}{|O|\cdot(D_o+1)}
=U_{EAI}(o).
\end{array}
\nonumber
\end{equation}}
\end{proof}


\iftacode
\begin{algorithm}[tb]
\algsetup{linenosize=\footnotesize}
  \footnotesize
\caption{Task Assignment}
\label{TaskAssignment}
\begin{algorithmic}[1]
 \renewcommand{\algorithmicrequire}{\textbf{Input:}}
 \REQUIRE set of workers $W$, number of questions $k$
	\STATE Compute the upper bound $U_{EMCI}(o)$ for $o \in O$
	\STATE $h_{UB}$ $\leftarrow$ BuildMaxHeap(\{$ \langle U_{EAI}(o),o \rangle $|$o \in O$\})
	\STATE Sort workers in the decreasing order of $\psi_{w,1}$ \\(i.e., $\psi_{1,1} \!\geq\! \psi_{2,1} \!\geq\! \cdots \!\geq\! \psi_{|W|,1}$).
	\FOR {$w=1$ to $|W|$}
		\STATE $h_{EAI}[w]$ $\leftarrow$ BuildMinHeap(\{\})
	\ENDFOR
	\WHILE {True}
		\STATE $\langle U_{EAI}(o),o \rangle \leftarrow h_{UB}$.extractMax()
		\IF {$h_{EAI}[|W|].size \eqmid k$ $\AND$ $h_{EAI}[w].min \!>\!  U_{EAI}(o)$ for all $w$}
			\STATE {\bf break}
		\ENDIF
		\FOR {$w=1$ to $|W|$}
			\IF{$w$ already answered on $o$ \textbf{or} $h_{EAI}[w].min \!>\!  U_{EAI}(o)$}
				\STATE {\bf continue} 
			\ENDIF
			\STATE Compute $EAI(w,o)$
			\STATE $h_{EAI}[w]$.insert($ \langle EAI(w,o),o \rangle $)
			\IF{$h_{EAI}[w].size \leq k$} 
				\STATE {\bf break}
			\ENDIF
			\STATE $o \leftarrow h_{EAI}[w]$.extractMin().value()
		\ENDFOR	
	\ENDWHILE
\end{algorithmic}
\end{algorithm}
\fi

We devise an algorithm to assign the best $k$ objects to each available worker in crowdsourcing systems.
Since a single answer is sufficient to find the correct value for some objects,
we assign an object to only a single worker in each round.
If the answer is not sufficient to find the correct value of the object, we assign the object to another worker in the next round.

Our task assignment algorithm sequentially assigns each object to a worker by scanning the objects $o$ with non-increasing order of the upper bound $U_{EAI}(o)$.
To allocate an object to a worker, since $\psi_{w,1}$ is the probability of answering the truth, we consider the workers $w$ with non-increasing order of $\psi_{w,1}$.
After assigning an object to a worker $w$, if the number of assigned objects to the worker $w$ exceeds $k$, we remove the object $o$ with the minimum $EAI(w,o)$ and assign the deleted object to the next worker and perform the same step.
While scanning the objects, we stop the assignment if the upperbound $U_{EAI}(o)$ is smaller than the minimum $EAI(w,o')$ among the $EAI(w,o')$s of all assigned objects and each worker has $k$ assigned objects.
The reason is that the $EAI(w,o)$ of the remaining objects $o$ can be larger than that of any assigned object.
\iftacode

\minisection{The pseudocode} 
It is shown in Algorithm~\ref{TaskAssignment}.
We first compute the upper bound $U_{EAI}(o)$ for every object $o \in O$ by Lemma \ref{lem:UBEAI} and build a maxheap $h_{UB}$ of all objects by using $U_{EAI}(o)$ as the key to assign the objects to workers in the decreasing order of $U_{EAI}(o)$ (in lines 1-2).
The workers are sorted in the decreasing order of $\psi_{w,1}$ to give a higher priority to reliable workers (in line 3). 
We next initialize a minheap $h_{EAI}[w]$ for each worker $w$ to contain the $k$ assigned objects (in lines 4-6).
Then, we repeatedly extract an object from $h_{UB}$ and assign the object to a worker in the sorted order of $\psi_{w,1}$ (in lines 12-18). 
Before assigning an object $o$, if the heaps $h_{EAI}[w]$s of all workers are full and the minimum value of $EAI(w,o')$ of the objects $o'$ in all $h_{EAI}[w]$s is larger than the upper bound $U_{EAI}(o)$, we stop immediately. 
\else
Due to the lack of space, we omit the pseudocode of the algorithm.
\fi

\section{Experiments}
\label{sec:experiment}


\eat{
\subsection{Test Environments}
\label{sec:test_env}}
The experiments are conducted on a computer with Intel i5-7500 CPU and 16GB of main memory.
 
\minisection{Datasets}
We collected the two real-life datasets publicly available at \url{http://kdd.snu.ac.kr/home/datasets/tdh.php}.
\ifpublishing
\else
\fi

\birthplaces: We crawled 13,510 records about the birthplaces of 6,005 celebrities from 7 websites (sources).
For the gold standard data to evaluate the correctness of discovered birthplaces, we used IMDb biography which is available at {\em http://www.imdb.com}.
Moreover, the geographical hierarchy was created by using the IMDb data.
For example, if there is a person who was born in `LA, California, USA',
we assigned `LA' as a child of `California' and `California' as a child of `USA'.
The hierarchy contains 4,999 nodes (e.g., countries, cities and etc.) and its height is 5.


\heritages: This is a dataset of the locations of World Heritage Sites provided by UNESCO World Heritage Centre, available at {\em http://whc.unesco.org}. 
We queried about the locations of 785 World Heritage Sites with Bing Search API and obtained 4,424 claimed values from 1,577 distinct websites.
The hierarchy was created in the same way as we did for \birthplaces and it has 1,027 nodes.
The height of this hierarchy tree is 6.

\eat{We use \acc, \accgen and \avgdist to evaluate the truth discovery algorithms.
	Let $t_o$ be the truth of the object $o$ in the gold standard.
	\acc is the ratio of objects that the algorithm discovers the truth exactly (i.e., $v_o^*\eqmid t_o$).
	\accgen is the ratio of the objects that the estimated truth belongs to the set of generalized truths (i.e., $v_o^* \inmid G_o(t_o) \midsize{\cup} \{t_o\}$).
	\avgdist is the average distance between the truth $t_o$ and the estimated truth $v_o^*$ in the hierarchy $H$ where the distance is the number of edges between them.
	That is, 

	\begin{equation*}
	\label{eq:measures}
	\begin{aligned}
	(\acc) &= \frac{\sum_{o\in O}{I(v_o^*\eqmid t_o)}}{|O|},\\
	(\accgen) &= \frac{\sum_{o\in O}{I(v_o^* \inmid G_H(t_o) \!\cup\! \{t_o\})}}{|O|},
	\end{aligned}
	\end{equation*}
	\begin{equation*}
	\begin{aligned}
	(\avgdist) &= \frac{\sum_{o\in O}{d(v_o^*,t_o)}}{|O|}
	\end{aligned}
	\end{equation*}
	where $d(u,v)$ is the distance between two nodes in the hierarchy and $I(condition)$ is the indicator function which returns 1 if the $condition$ is true and 0 otherwise.
}

\minisection{Quality Measures}
	We use \acc, \accgen and \avgdist to evaluate the truth discovery algorithms.
	Let $t_o$ be the truth of the object $o$ in the gold standard and $v_o^*$ be the estimated truth by an algorithm.
	Note that $t_o$ may not exist in the set of candidate values.
	In this case, the most specific candidate value among the ancestors of the truth is assumed to be $t_o$.
	\acc is the proportion of objects that the algorithm discovers the truth exactly.
	It is actually used in \cite{DOCS,dong2015knowledge,zheng2015qasca,zheng2017truth} to evaluate truth discovery algorithms.
	\begin{equation*}
	\ifshrinkmode
	(\acc) = {\Sigma_{o\in O}{I(v_o^* = t_o)}}/{|O|}
	\fi
	\end{equation*}
	The ancestors of $t_o$ are less informative but still correct values.
	Thus, we develop an evaluation measure named \accgen which is the proportion of objects $o$ whose estimated truth $v_o^*$ is either the truth $t_o$ or an ancestor of the truth.
	\begin{equation*}
	\ifshrinkmode
	(\accgen) = {\Sigma_{o\in O}{I(v_o^* \inmid G_H(t_o) \!\cup\! \{t_o\})}}/{|O|}
	\fi
	\end{equation*}
	
	Ancestors of the truth have a different level of informativeness depending on the distance to the truth.
	For example, `New York' is more informative than `USA' as the location of the Statue of Liberty. 
	Thus, we utilize another evaluation measure named \avgdist which weights the estimated truth based on the distance from the ground truth.
	More specifically, it is the average number of edges $d(v_o^*,t_o)$ between the truth $t_o$ and the estimated truth $v_o^*$ in the hierarchy $H$.
	\begin{equation*}
	\ifshrinkmode
	(\avgdist) = {\Sigma_{o\in O}{d(v_o^*,t_o)}}/{|O|}
	\fi
	\end{equation*} 
	\avgdist is robust to the case where the ground truth is less specific than the estimated truth.
	The estimated truth is regarded as a wrong value when we compute \acc and \accgen even though the estimate truth is correct and more specific.
	Since the distance between the less specific ground truth and the estimated truth is generally small, \avgdist compensates the drawback of \acc and \accgen.
	
	\eat{
		Since the distance is small, note that \avgdist compensates the drawback of \acc and \accgen.
		
		Since the distance between the ground truth and the estimated truth is small, note that \avgdist compensates the drawback of \acc and \accgen.}

\begin{table}[tb]
	\center
	\caption{Performance of truth inference algorithms}
	\scriptsize
	\vspace{-0.08in}
	\label{tab:ti}
	\setlength{\tabcolsep}{0.47em}
	\begin{tabular}{c|ccc|ccc}
		\toprule
		& \multicolumn{6}{c}{Dataset}\\
		\midrule
		& \multicolumn{3}{c|}{\birthplaces} & \multicolumn{3}{c}{\heritages}\\
		Algorithm & \acc & \accgen & \avgdist & \acc & \accgen & \avgdist \\
		\midrule
		TDH & \bf{0.8913} & \bf{0.8988} & \bf{0.3151} & \bf{0.7414} & 0.8726 & \bf{0.5210}\\
		VOTE & 0.7900 & 0.8924 & 0.4961 & 0.6892 & \bf{0.8994} & 0.6382\\
		LCA & 0.8834 & 0.8923 & 0.3414 & 0.6930 & 0.8866 & 0.6611\\
		DOCS & 0.8828 & 0.8916 & 0.3409 & 0.6904 & 0.8866 & 0.6599\\
		ASUMS & 0.8543 & 0.8571 & 0.4573 & 0.6229 & 0.7414 & 1.2000\\
		MDC & 0.8263 & 0.8432 & 0.5320 & 0.7254 & 0.8087 & 0.6869\\
		ACCU & 0.8137 & 0.8296 & 0.6063 & 0.5834 & 0.7656 & 1.0637\\
		POPACCU & 0.8133 & 0.8300 & 0.6070 & 0.6561 & 0.8586 & 0.7554\\
		LFC & 0.8085 & 0.8743 & 0.4669 & 0.6803 & 0.8076 & 0.8076\\
		CRH & 0.8083 & 0.8271 & 0.6120 & 0.6841 & 0.8828 & 0.6688\\
		\bottomrule
	\end{tabular}
	\vspace{-0.12in}
\end{table}

\minisection{Settings for simulated crowdsourcing}
To evaluate the truth discovery algorithms with varying the quality of the answers from workers, we conducted experiments with simulated crowd workers.
In our simulation, we assumed that each simulated worker answers a question correctly with its own probability $p_w$ and randomly selects an answer from the candidate values with probability $1$$-$$p_w$. 
We sampled the probability $p_w$ from a uniform distribution ranging from $\pi_p$$-$$0.05$ to $\pi_p$$+$$0.05$ where the default value of $\pi_p$ is 0.75. 
In the experiments, each of 10 worker answers 5 questions for each round.

\subsection{Implemented Algorithms}
We implemented 10 truth inference algorithms and 4 task assignment algorithms in Python for comparative experiments.
The truth inference algorithms are referred to as follows:

\begin{itemize}[topsep=1pt,itemsep=0ex,partopsep=1ex,parsep=0.5ex, leftmargin=.2in]
\ifdomain
\else
\item{TDH}: This is our algorithm proposed in Section~\ref{TruthInference}. 
For the prior distribution $Dir(\alpha)$, we set the hyperparameter $\alpha=(3,3,2)$  since correct values are more frequent than wrong values for most of the sources. For the other hyperparameters $\beta$ and $\gamma$, we set every dimension of $\beta$ and $\gamma$ to 2. 
\fi
\item{ACCU}: It is the algorithm proposed in \cite{dong2009integrating} which considers the dependencies between sources to find the truths. The algorithm exploits Bayesian analysis to find the dependencies.
\item{POPACCU}: This denotes the algorithm in \cite{dong2012less} which extends ACCU. It computes the distribution of the false values from the records while ACCU assumes that it is uniform.
\item{LFC}: This algorithm is proposed in \cite{raykar2010learning} and utilizes a confusion matrix to model a source's quality. 
\item{CRH}: It is proposed in \cite{li2014resolving} to resolve conflicts in heterogeneous data containing categorical and numerical attributes.
\item{LCA}: It is a probabilistic model proposed in \cite{LCA}. We select GuessLCA to be compared in this paper which is one of the best performers among the 7 algorithms proposed in \cite{LCA}. 
\item{ASUMS}: This is proposed in \cite{asums} by adapting an existing method SUMS \cite{sums} to hierarchical truth discovery. 
\item{MDC}: This denotes the truth discovery method designed for medical diagnose from non-expert crowdsourcing in \cite{li2017reliable}. 
\item{DOCS}: This is the state-of-the-art technique presented in \cite{DOCS} that suggests the domain-sensitive worker model.

\item{VOTE}: This is a baseline that selects a value with the highest frequency in the claimed values. 
\vspace{-0.03in}
\end{itemize}
\vspace{0.03in}
We implemented the following task assignment algorithms.
\begin{itemize}[topsep=1pt,itemsep=0ex,partopsep=0ex,parsep=0.5ex, leftmargin=.2in]
\item \emph{EAI}: This is our proposed algorithm in Section \ref{sec:task_assignment}.
\item \emph{MB}: It is the task assignment algorithm used by DOCS \cite{DOCS}.
\item \emph{QASCA}: It is a task assignment algorithm proposed in \cite{zheng2015qasca}. 

\item \emph{ME}: This is our baseline algorithm which utilizes an uncertainty sampling.
It selects an object $o^*$ whose confidence distribution has the maximum entropy.
(i.e., $o^* = argmax_{o\in O}$ ${(-\sum_{v\in V_o}{\mu_{o,v}\cdot \log{\mu_{o,v}}})}$)
\end{itemize}

Note that \emph{EAI} and \emph{MB} are the task assignment algorithms specially designed to work with \emph{TDH} and \emph{DOCS}, respectively.
\emph{QASCA} can work with truth inference algorithms based on probabilistic models such as \emph{TDH}, \emph{DOCS}, \emph{LCA}, \emph{ACCU} and \emph{POPACCU}.
All the truth inference algorithms can be combined with \emph{ME}. 


\begin{figure}[tb]
	\centering
	\includegraphics[width=3.3in]{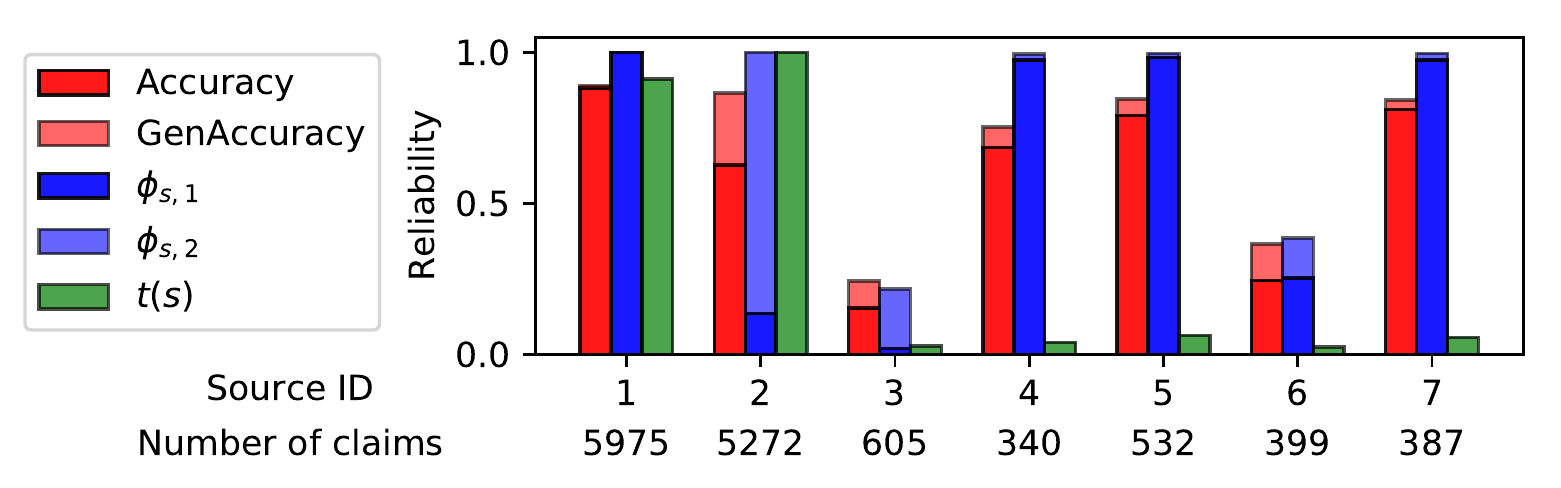}
	\vspace{-0.13in}
	\caption{Source reliability distribution in \birthplaces \label{exp:source_distribution_bar}}
	\vspace{-0.1in}
\end{figure}

\subsection{Truth Inference}
\label{truth_inference}
We first provide the performances of the truth inference algorithms without using crowdsourcing in Table.~\ref{tab:ti}.

\minisection{\birthplaces}
\eat{The average accuracy of the sources in \birthplaces is 72.1\% and every algorithm shows higher \acc than 72.1\%.}
Our TDH outperforms all other algorithms in terms of all quality measures since
TDH finds the exact truths by utilizing the hierarchical relationships. 
Since TDH estimates the reliabilities of the sources and workers by considering the hierarchies, it does not underestimate the reliabilities of the sources and workers.
Thus, TDH also finds more correct values including the generalized truths.
We will discuss the reliability estimation in detail at the end of this section by comparing TDH with ASUMS.
LCA is the second-best performer and VOTE shows the lowest \acc among all compared algorithms.
However, in terms of \accgen, VOTE performs the second-best.
It is because many websites claim the generalized values rather than the most specific value.

\minisection{\heritages}
In terms of \avgdist and \acc, TDH performs the best among those of the compared algorithms.
VOTE shows the highest \accgen because many sources provide the generalized truths.
In fact, a high \accgen with low \acc and \avgdist can be easily obtained by providing the most general values for the truths.
However, such values usually are not informative.
Since our algorithm shows much higher \acc and much lower \avgdist than VOTE, we can see that the estimated truth by TDH is more accurate and precise than the result from VOTE.
\heritages contains many sources and most of the sources have a few claims.
Thus, it is very hard to estimate the reliability of each source accurately.
Therefore, most of the compared algorithms show worse performance than VOTE in terms of \avgdist.
In particular, ACCU has the lowest \acc. 
The reason is that ACCU requires many shared objects between two sources in order to accurately determine the dependency between the sources.
The average accuracy of the sources in \heritages is 58.0\% while that of the sources in \birthplaces is 72.1\%. Thus, every algorithm shows a lower \acc in this dataset than in \birthplaces. 

%


\begin{figure}[tb]
	\vspace{-0.06in}
	\centering
	\subfloat[\birthplaces]{\includegraphics[width=1.6in]{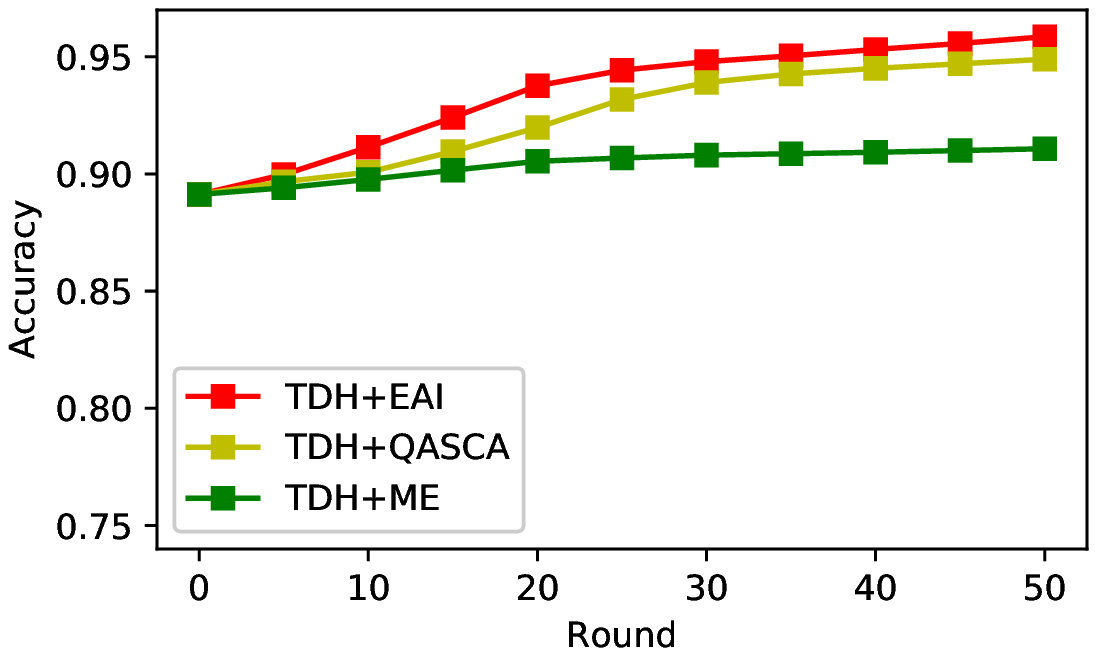}%
		\label{exp:accu_bp_ta}}
	\subfloat[\heritages]{\includegraphics[width=1.6in]{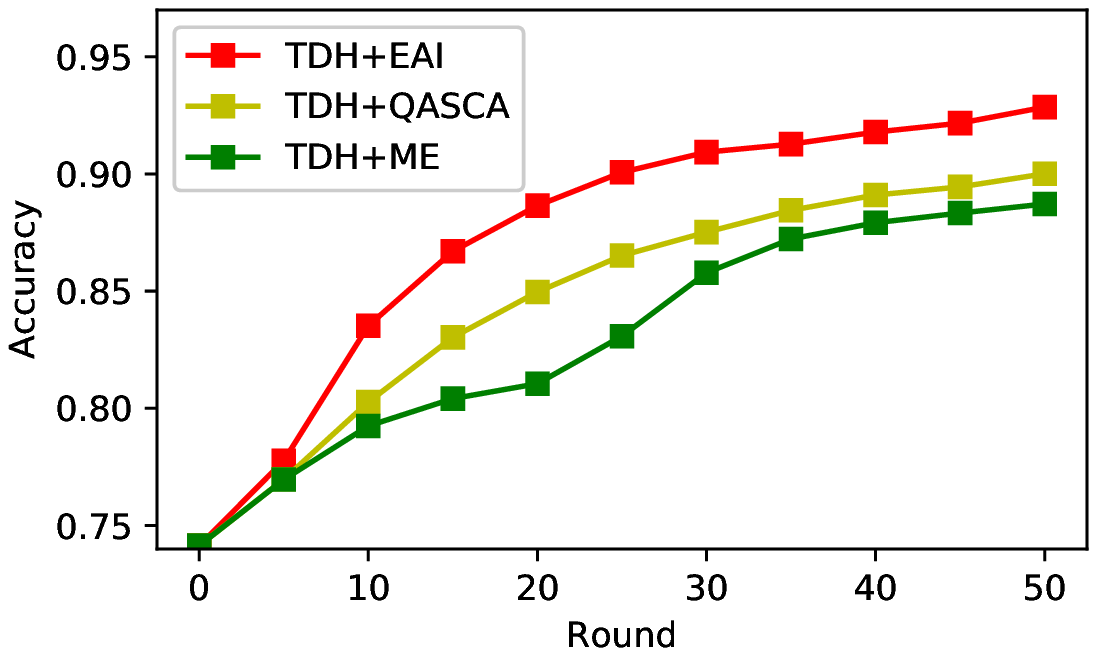}%
		\label{exp:accu_whc_ta}}
	\vspace{-0.05in}
	\caption{Evaluation of task assignment algorithms}
	\label{exp:crowd_accuracy_ta}
	\vspace{-0.05in}
\end{figure}

\begin{figure*}[tb]
	\vspace{-0.12in}
	\centering
	\subfloat[\birthplaces -EAI]{\includegraphics[width=1.4in]{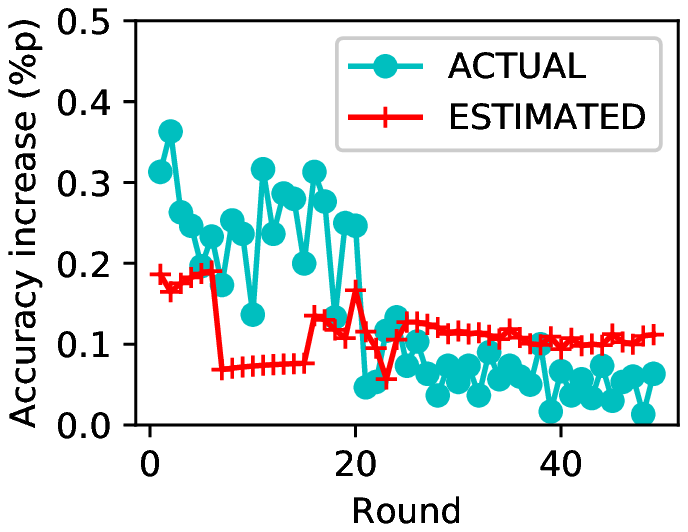}%
		\label{exp:compta_bp_eai}}
	\hspace{0.2in}
	\subfloat[\birthplaces -QASCA]{\includegraphics[width=1.4in]{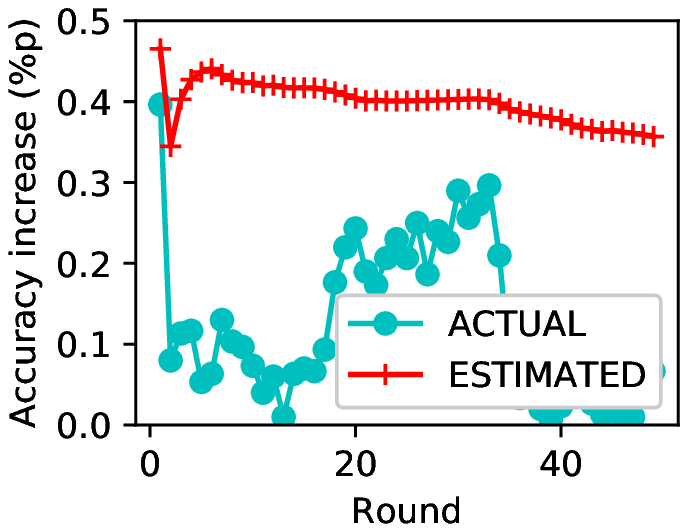}%
		\label{exp:compta_bp_qasca}}
	\hspace{0.2in}
	\subfloat[\heritages -EAI]{\includegraphics[width=1.35in]{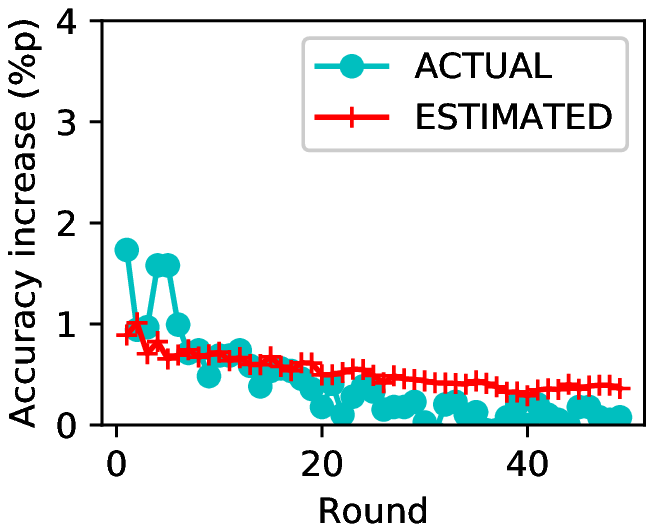}%
		\label{exp:compta_whc_eai}}
	\hspace{0.2in}
	\subfloat[\heritages -QASCA]{\includegraphics[width=1.35in]{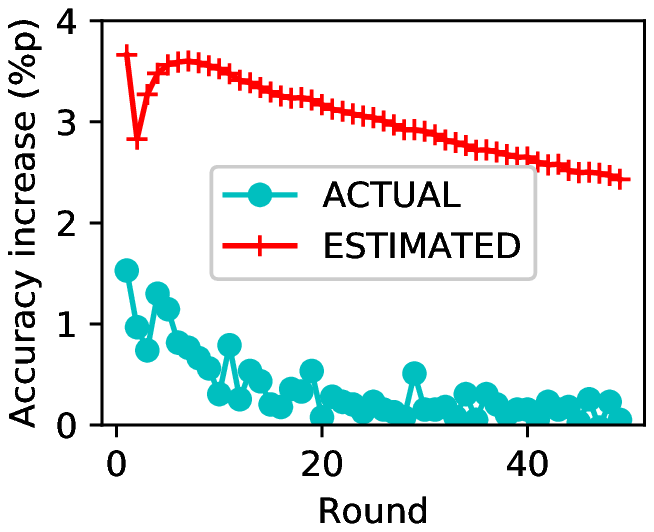}%
		\label{exp:compta_whc_qasca}}
	\vspace{-0.1in}
	\caption{Actual and estimated accuracy improvement by EAI and QASCA}
	\label{exp:compta}
	\vspace{-0.1in}
\end{figure*}

\newif\ifctdtwocolumn
\ctdtwocolumntrue

\ifctdtwocolumn
\begin{figure}[tb]
	\centering
	\subfloat[\birthplaces]{\hspace{-0.03in}
		\includegraphics[width=1.34in]{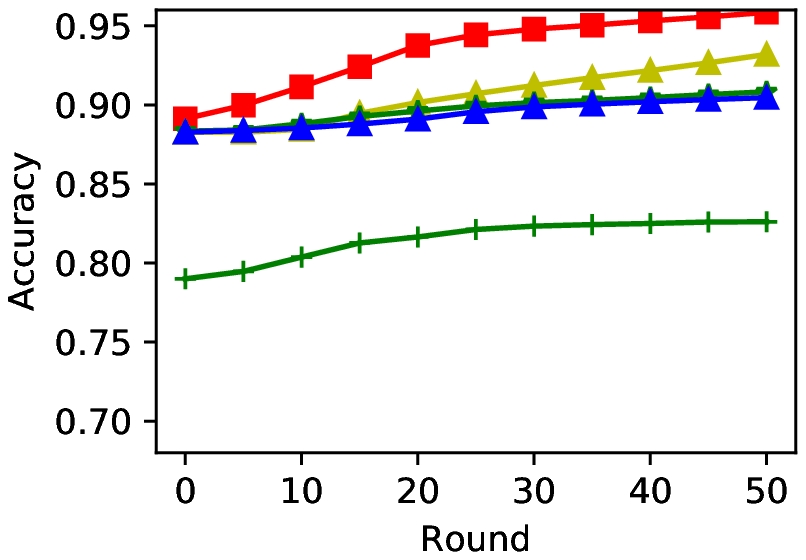}%
		\label{exp:accu_bp}}
	\subfloat[\heritages]{\hspace{-0.07in}
		\includegraphics[width=1.99in]{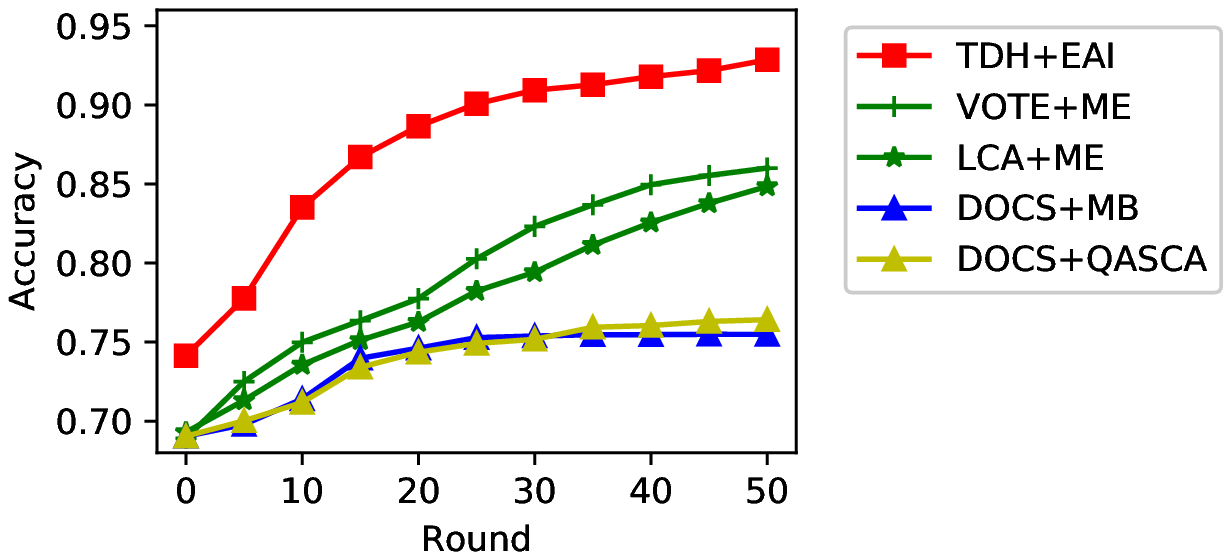}%
		\label{exp:accu_whc}}
	\vspace{-0.09in}
	\caption{\acc with crowdsourced truth discovery}
	\label{exp:crowd_accuracy}
	\vspace{-0.11in}
\end{figure}
\begin{figure}[tb]
	\vspace{-0.07in}
	\centering
	\subfloat[\birthplaces]{\includegraphics[width=1.33in]{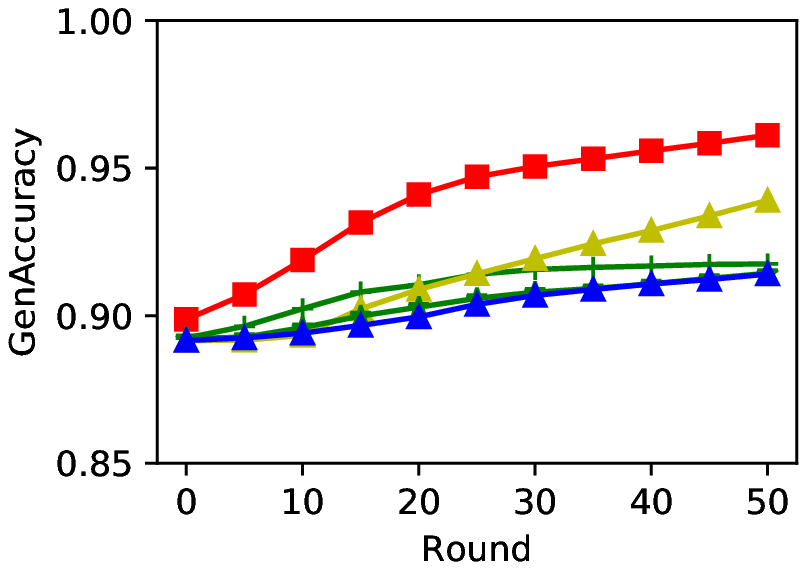}%
		\label{exp:accugen_bp}}
	\subfloat[\heritages]{\hspace{-0.05in}\includegraphics[width=2.00in]{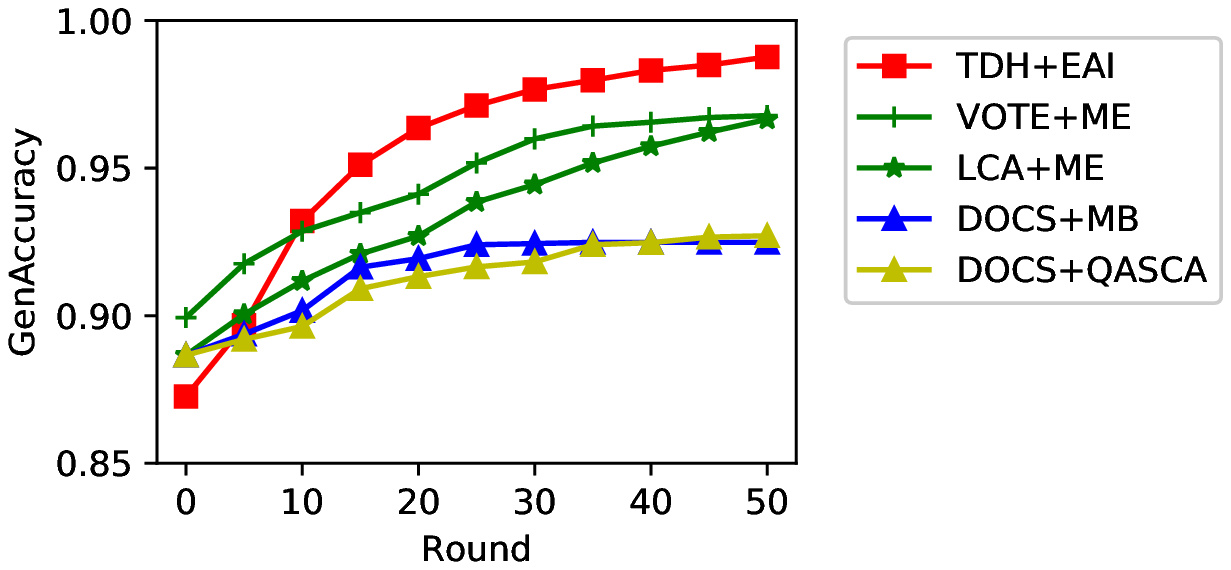}%
		\label{exp:accugen_whc}}
	\vspace{-0.1in}
	\caption{\accgen with crowdsourced truth discovery}
	\label{exp:crowd_accuracygen}
	\vspace{-0.13in}
\end{figure}
\begin{figure}[t!]
	\vspace{-0.07in}
	\centering
	\subfloat[\birthplaces]{\includegraphics[width=1.33in]{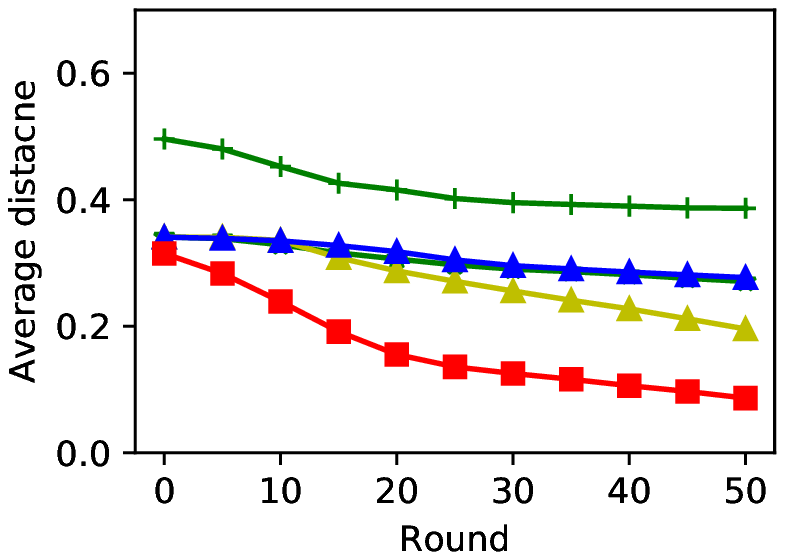}%
		\label{exp:avgdist_bp}}
	\subfloat[\heritages]{\hspace{-0.05in}\includegraphics[width=2.0in]{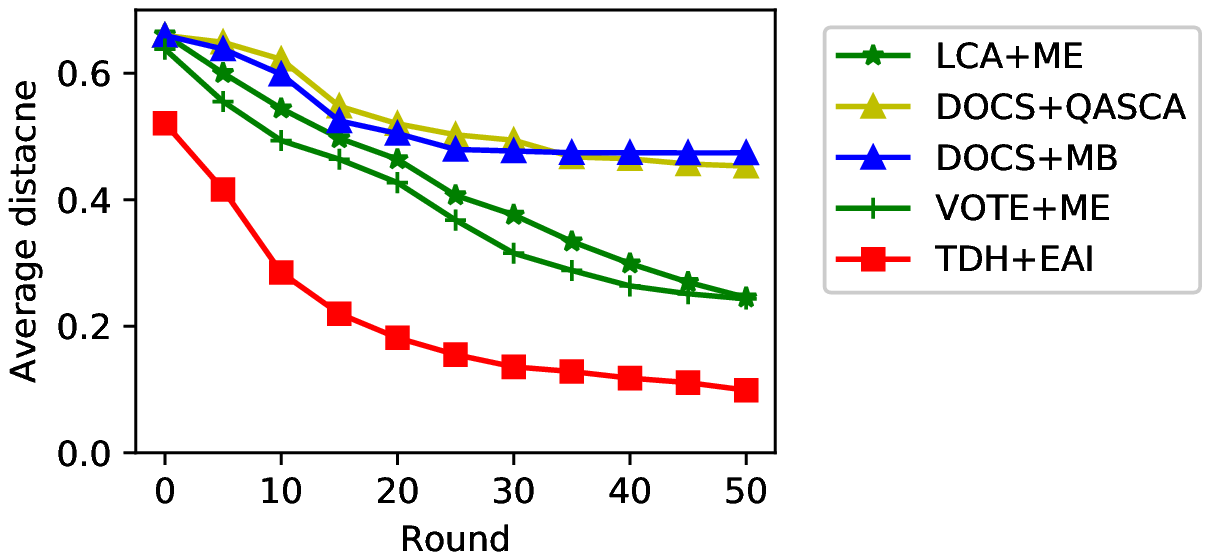}%
		\label{exp:avgdist_whc}}
	\vspace{-0.1in}
	\caption{\avgdist with crowdsourced truth discovery}
	\label{exp:crowd_avgdist}
	\vspace{-0.1in}
\end{figure}
\else
\begin{figure*}[hb]
	\centering
	\subfloat[\acc]{\includegraphics[width=1.75in]{birthPlaces_q5_w10_lb0_7_ub0_8_scF/Accuracy_MAX_ENTROPY_nolegend}%
		\label{exp:accu_bp}}
	\subfloat[\accgen]{\includegraphics[width=1.75in]{birthPlaces_q5_w10_lb0_7_ub0_8_scF/AccuracyGen_MAX_ENTROPY_nolegend}%
		\label{exp:accugen_bp}}
	\subfloat[\avgdist]{\includegraphics[width=2.5in]{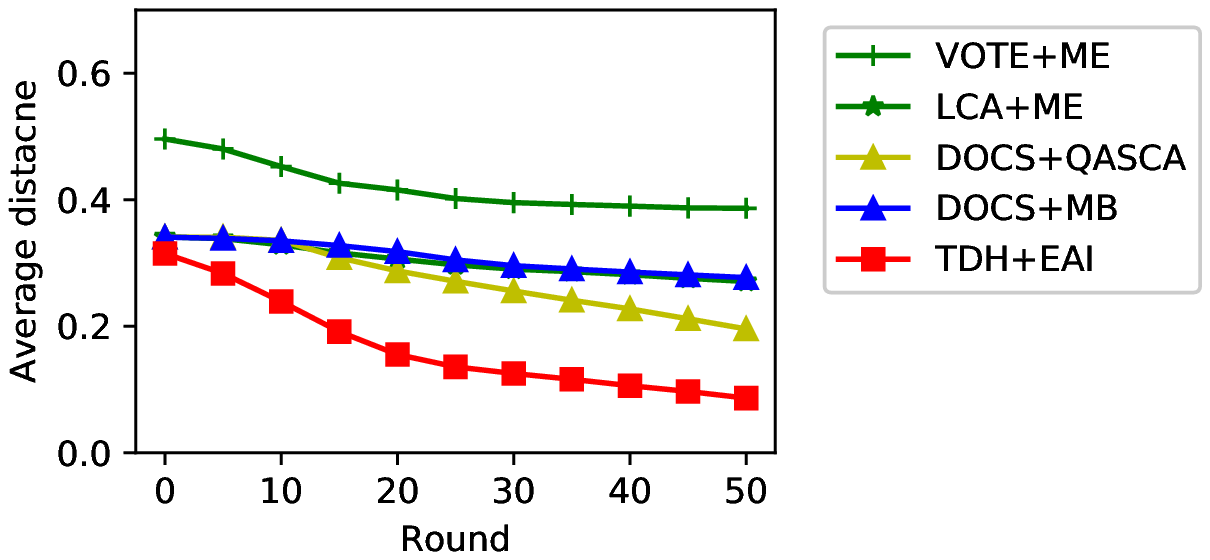}%
		\label{exp:avgdist_bp}}
	\caption{Crowdsourced truth discovery in \birthplaces}
	\label{exp:crowd_bp}
\end{figure*}
\begin{figure*}[hb]
	\centering
	\subfloat[\acc]{\includegraphics[width=1.75in]{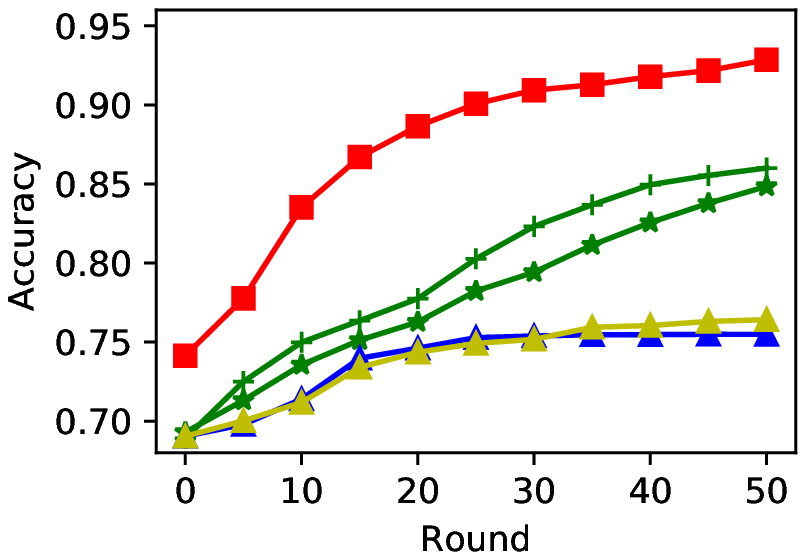}%
		\label{exp:accu_whc}}
	\subfloat[\accgen]{\includegraphics[width=1.75in]{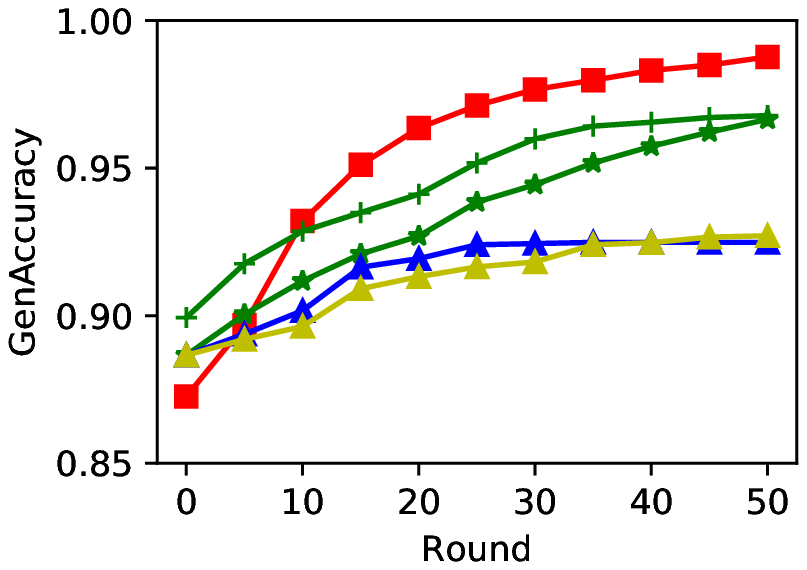}%
		\label{exp:accugen_whc}}
	\subfloat[\avgdist]{\includegraphics[width=2.5in]{whc_q5_w10_lb0_7_ub0_8_scF/averagedistance_MAX_ENTROPY}%
		\label{exp:avgdist_whc}}
	\caption{Crowdsourced truth discovery in \heritages \label{exp:crowd_whc}}
\end{figure*}
\fi

\minisection{Comparison with ASUMS}
Since ASUMS \cite{asums} is the only existing algorithm which utilizes hierarchies for truth inference, we show the statistics related to the reliability distributions estimated by TDH and ASUMS for \birthplaces dataset in \figurename~\ref{exp:source_distribution_bar}.
\acc and \accgen represent the actual reliabilities of each source computed from the ground truths.
Recall that $\phi_{s,1}$ and $\phi_{s,2}$ are the estimated probabilities of providing a correct value and a generalized correct value respectively for a source $s$ by our TDH, as defined in Section~\ref{TruthInference}.
In addition, $t(s)$ is the estimated reliability of a source $s$ by ASUMS which ignores the generalization level of each source.
The reliabilities of the sources 4, 5 and 7 computed by ASUMS (i.e. $t(s)$) are quite different from the actual reliabilities (i.e., \acc).
As we discussed in Section~\ref{intro}, for a pair of sources that provide different claimed values with an ancestor-descendant relationship in a hierarchy, existing methods may assume that one of the claimed values is incorrect. 
Thus, the reliability of the source with the assumed wrong value tends to become lower by the existing methods.
ASUMS suffers from the same problem and underestimates the reliabilities of the sources 4, 5 and 7 which provide a small number of claimed values. 
Meanwhile, our proposed algorithm TDH accurately estimates the reliabilities of the sources by introducing another class of the claimed values (generalized truth). 

\begin{table}[b]
	\center
	\vspace{-0.05in}
	\caption{\acc of the algorithms after the 50th round}
	\label{tab:crowdsourced50}
	\scriptsize
	\vspace{-0.03in}
	\setlength{\tabcolsep}{0.6em}
	\def\arraystretch{1.0}
	\begin{tabular}{c|cccc|cccc}
		\toprule
		& \multicolumn{4}{c}{\birthplaces}& \multicolumn{4}{c}{\heritages}\\
		\midrule
		& EAI & MB & QASCA & ME& EAI & MB & QASCA & ME\\
		\midrule
		TDH & {\bf 0.9601} & -  & {\bf 0.9500} & {\bf 0.9109} & {\bf 0.9304} & -  & {\bf 0.8999} & {\bf 0.8884}\\
		DOCS & -  & {\bf 0.9052} & \underline{0.9341} & 0.8842 & -  & {\bf 0.7546} & \underline{0.7661} & 0.7631\\
		LCA & -  & -  & 0.8823 & \underline{0.9089} & -  & -  & 0.7136 & 0.8507\\
		POPACCU & -  & -  & 0.9295 & 0.8987 & -  & -  & 0.7512 & 0.8336\\
		ACCU & -  & -  & 0.8468 & 0.8257 & -  & -  & 0.5796 & 0.5896\\
		ASUMS & -  & -  & -  & 0.8700 & -  & -  & -  & 0.7427\\
		CRH & -  & -  & -  & 0.9000 & -  & -  & -  & 0.8459\\
		MDC & -  & -  & -  & 0.8254 & -  & -  & -  & 0.7241\\
		LFC & -  & -  & -  & 0.8287 & -  & -  & -  & 0.7327\\
		VOTE & -  & -  & -  & 0.8261 & -  & -  & -  & \underline{0.8634}\\
		\bottomrule
	\end{tabular}
	
	\vspace{-0.05in}
\end{table}

%

\subsection{Task Assignment}

Before providing the full comparison of all possible combinations of truth inference algorithms and task assignment algorithms,
we first evaluate the task assignment algorithms with our proposed truth inference algorithm.
We plotted the average \acc of the truth discovery algorithms with different task assignment algorithms for every 5 round in \figurename~\ref{exp:crowd_accuracy_ta}.
The points at the 0-th round represent the \acc of the algorithms without crowdsourcing.
All algorithms show the same \acc at the beginning since they use the same truth inference algorithm TDH.
As the round progresses, the \acc of TDH+EAI increases faster than those of all other algorithms.
The \acc of TDH+ME is the lowest since ME selects a task based only on the uncertainty without estimating the accuracy improvement by the task.

As discussed in Section~\ref{sec:quality_measure}, our task assignment algorithm EAI estimates the accuracy improvement by considering the number of existing claimed values and the confidence distribution whereas QASCA considers the confidence distribution only.
We plotted the actual and estimated accuracy improvements by EAI and QASCA in \figurename~\ref{exp:compta}.
The graphs show that the estimated accuracy improvement by EAI is similar to the actual accuracy improvement while QASCA overestimates the accuracy improvement at every round.
On average, the absolute estimation errors from EAI are 0.08 and 0.26 percentage points (pps) while those errors from QASCA are 0.28 and 2.66 pps in \birthplaces and \heritages datasets, respectively.
This result confirms that EAI outperforms QASCA by effectively estimating the accuracy improvement.
In terms of the other quality measures \accgen and \avgdist, our proposed EAI also outperforms the other task assignment algorithms in both datasets.
Due to the lack of space, we omit the results with the other quality measures.

\subsection{Simulated Crowdsourcing}
\label{sec:expCATD}
We evaluate the performance of crowdsourced truth discovery algorithms with the simulated crowdsourcing.

For all possible combinations of the implemented truth inference and task assignment algorithms, we show the \acc after 50 rounds of crowdsourcing in Table~\ref{tab:crowdsourced50} where the impossible combinations are denoted by `-'.
As expected, TDH+EAI has the highest \acc in both datasets for all possible combinations.
\eat{
Since TDH+EAI shows higher \acc than TDH+QASCA in both dataset, we can see that our task assignment algorithm effectively selects the task which is likely to increase the accuracy the most.} 
\red{
The result also shows that both TDH and EAI contribute to increasing \acc.
The improvement obtained by EAI can be estimated by comparing the result of TDH+EAI to that of the second performer TDH+QASCA.
The accuracies of TDH+EAI in \birthplaces and \heritages datasets are 1 and 3 percentage points (pps) higher than those of TDH+QASCA, respectively.
In addition, for each combined task assignment algorithm, the improvement by TDH can be inferred by comparing the results with those of other truth inference algorithms.
In both datasets, TDH shows the highest \acc among the applicable truth inference algorithms for each task assignment algorithm.
For example, TDH+QASCA shows 2.6 and 13 pps higher \acc in \birthplaces and \heritages datasets, respectively, than the second performer DOCS+QASCA among the combinations with QASCA.}
In the rest of the paper, we report \acc, \accgen and \avgdist of TDH+EAI, DOCS+MB, DOCS+QASCA, LCA+ME and VOTE+ME only since these combinations are the best or the second-best for each task assignment algorithm. 
%

\minisection{Cost efficiency}
\ifctdtwocolumn
We plotted the average \acc of the tested algorithms for every 5 rounds in Figure~\ref{exp:crowd_accuracy}.
\else
We plotted the average \acc of the tested algorithms for every 5 rounds in Figures~\ref{exp:crowd_bp} (a) and \ref{exp:crowd_whc}(a).
\fi
TDH+EAI shows the highest \acc for every round in both datasets.
For the \birthplaces dataset, DOCS+QASCA was the next best performer which achieved 0.9341 of \acc at the 50-th round.
Meanwhile, TDH+EAI only needs 17 rounds of crowdsourcing to achieve the same \acc.
Thus, TDH+EAI saved 66\% of crowdsourcing cost compared to the second-best performer DOCS+ QASCA.
Likewise, TDH+EAI reduced the crowdsourcing cost 74\% in \heritages dataset compared to the next performer. 
\ifctdtwocolumn
In terms of \accgen and \avgdist, TDH+EAI also outperforms all the other algorithms as plotted in Figure~\ref{exp:crowd_accuracygen} and Figure~\ref{exp:crowd_avgdist}.
\else
In terms of \accgen and \avgdist, TDH+EAI also outperforms all the other algorithms as plotted in Figure~\ref{exp:crowd_bp} and Figure~\ref{exp:crowd_whc}.
\fi
The results confirm that TDH+EAI is the most efficient as it achieves the best qualities in terms of both \acc and \accgen.

%


\minisection{Varying $\pi_p$}
We plotted the average \acc of all algorithms with varying the probability of correct answer $\pi_p$ of simulated workers for \birthplaces and \heritages datasets in Figure~\ref{exp:varyingMu}(a) and Figure~\ref{exp:varyingMu}(b), respectively.
As we can easily expect, the accuracies increase with growing $\pi_p$ for most of the algorithms.
For both datasets, TDH+EAI achieves the best accuracy with all values of $\pi_p$.
In Heritages dataset, a source provided less than 10 claims on average and it makes difficult for truth discovery algorithms to estimate the reliabilities of sources.
Therefore, the baseline VOTE+ME shows good performance on Heritages dataset.
Meanwhile, the performance of the state-of-the-art DOCS is significantly degraded on the Heritages dataset.

\begin{figure}[tb]
	\hspace{-0.05in}	
	\subfloat[\birthplaces]{\includegraphics[width=1.37in]{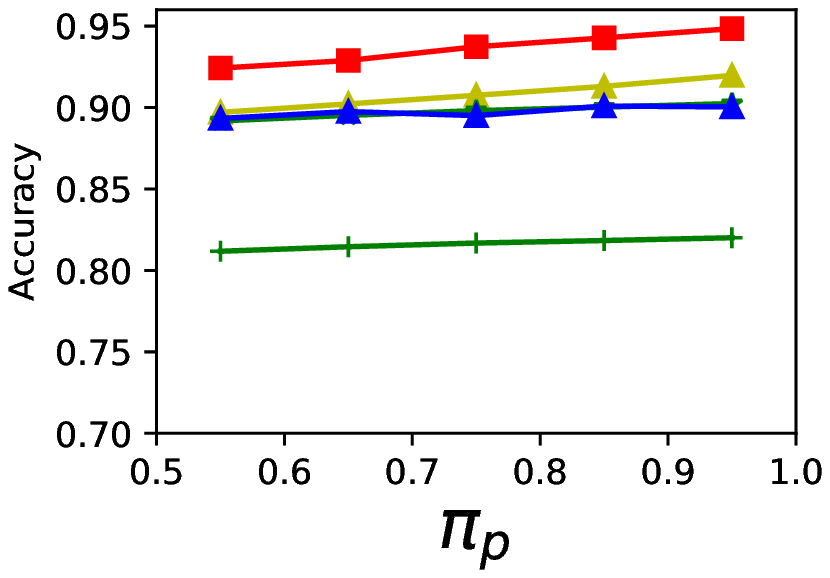}%
	\label{exp:varyingMuBP}}
	\subfloat[\heritages]{\hspace{-0.06in}\includegraphics[width=1.97in]{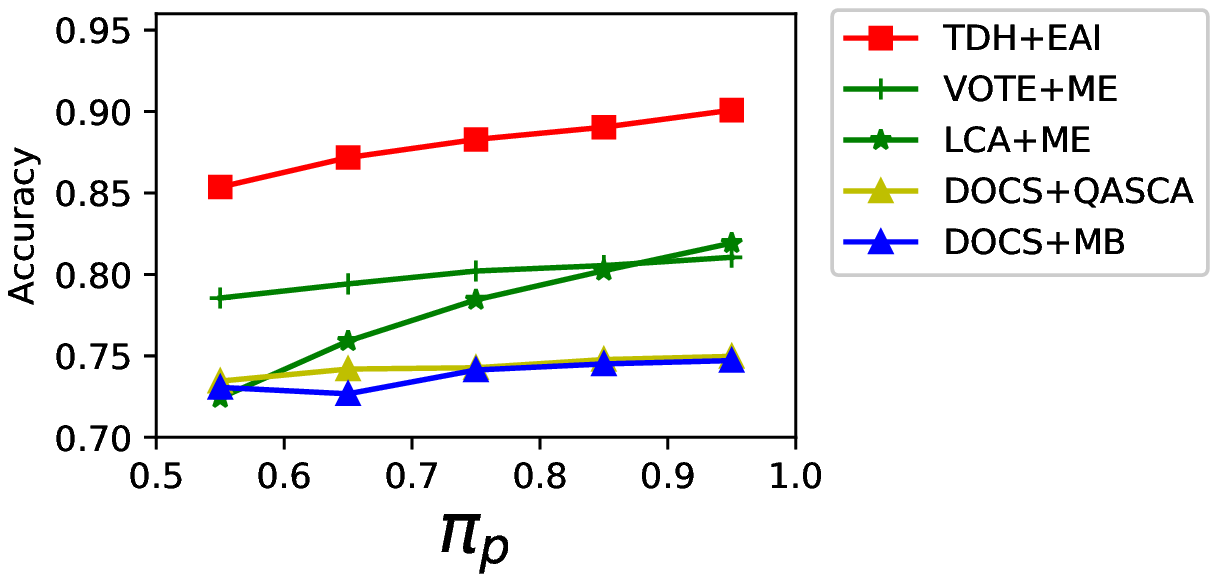}%
	\label{exp:varyingMuWHC}}
\vspace{-0.1in}
\caption{Varying $\pi_p$\label{exp:varyingMu}}
\vspace{-0.13in}
\end{figure}

\begin{figure}[tb]
\centering
\subfloat[\birthplaces]{\includegraphics[width=1.64in]{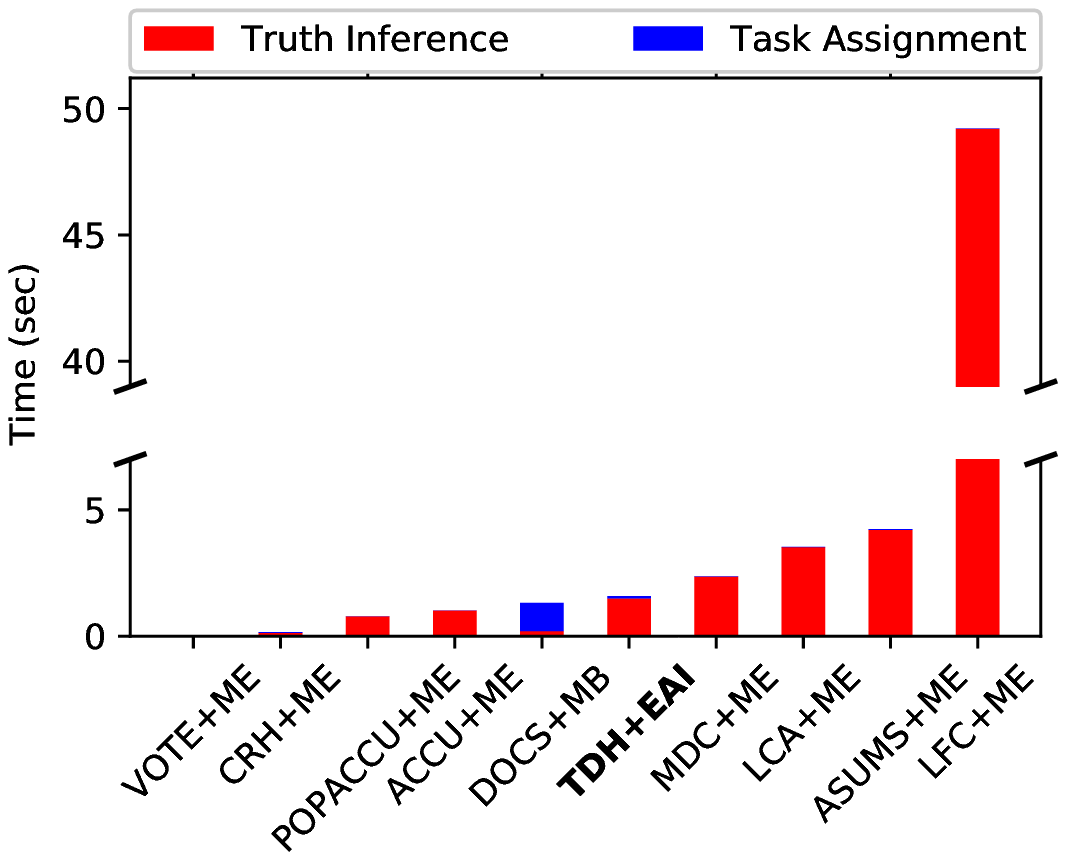}%
\label{exp:time_bp}}
\subfloat[\heritages]{\includegraphics[width=1.64in]{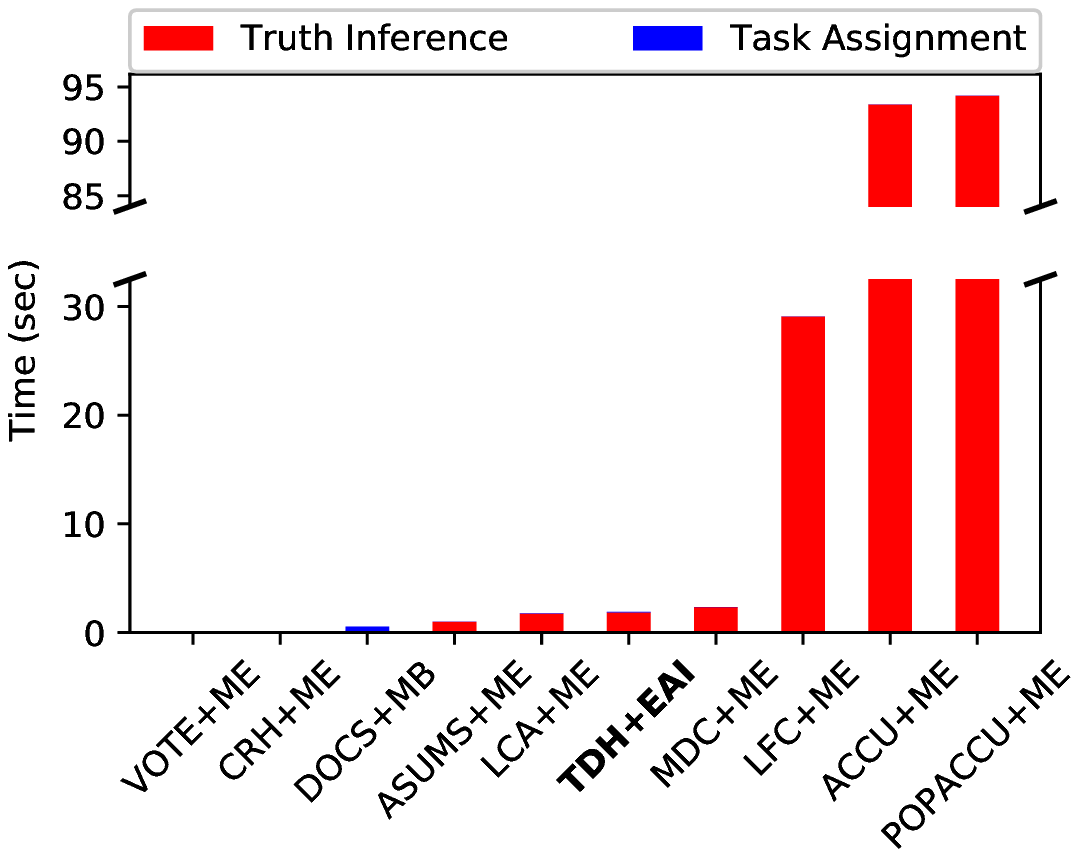}%
\label{exp:time_whc}}
\vspace{-0.1in}
\caption{Execution time per round\label{exp:execution_time}}
\vspace{-0.13in}
\end{figure}

\begin{figure}[t]
\centering
	\subfloat[\birthplaces]{\includegraphics[width=1.52in]{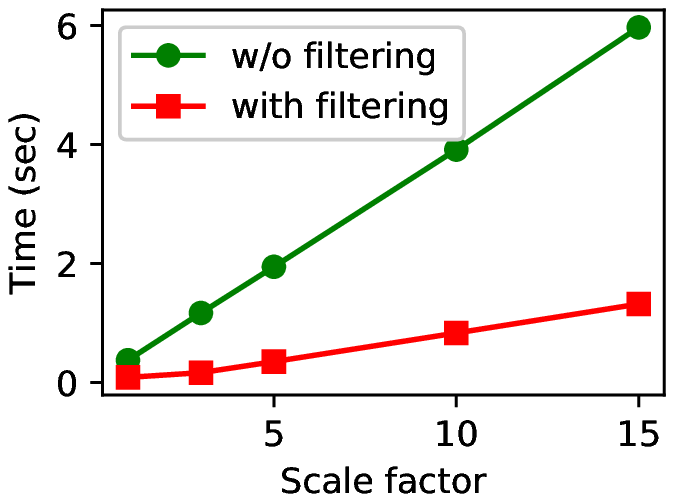}%
	\label{exp:scaleAssignBP}}
	\subfloat[\heritages]{\includegraphics[width=1.6in]{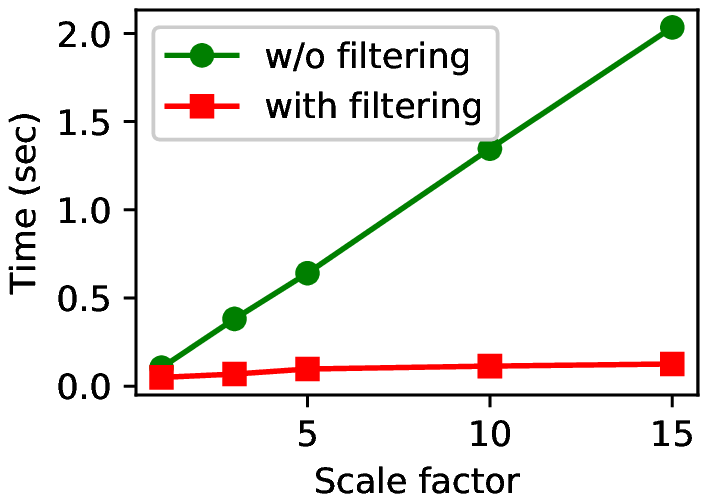}%
	\label{exp:scaleAssignWHC}}
\vspace{-0.06in}
\caption{Execution time for task assignment per round\label{exp:scaleAssign}}
\vspace{-0.1in}
\end{figure}

\minisection{Execution times} 
We plotted the average execution times of the tested algorithms over every round in Figure \ref{exp:execution_time}.
VOTE, CRH+ME, DOCS+MB and TDH+EAI run in less than 2.0 seconds per round on average for both datasets.
Other algorithms except for ACCU+ME, POPACCU+ME and LFC+ME also take less than 5 seconds, which is acceptable for crowdsourcing.
Since LFC builds the confusion matrix whose size is the square of the number of candidate values, LFC is the slowest with \birthplaces data.
On the other hand, for \heritages dataset which is collected from much more sources than \birthplaces dataset,
ACCU and POPACCU take longer time for truth inference to calculate the dependencies between sources.

\minisection{Effects of the filtering for task assignments}
To test the scalability of our algorithm, we increase the size of both datasets by duplicating the data by upto 15 times.
In Figure~\ref{exp:scaleAssign}, with increasing data size, we plotted the execution times of our task assignment algorithm EAI with and without exploiting the upper bound proposed in Section~\ref{sec:upperbound}.
The filtering technique saved 78\% and 94\% of the computation time for the task assignment at the scale factor 15. 
The graphs show that the proposed upper bound enables us to scale for large data effectively.
For the total execution time, including the truth inference, the filtering reduced 21\% and 6\% of the execution time on \birthplaces and \heritages respectively at the scale factor 15.

%
%
%

\newif\ifrealtwocol
\realtwocoltrue

\subsection{Crowdsourcing with Human Annotators}
\label{sec:expReal}

We evaluated the performance of the truth discovery algorithm by crowdsourcing real human annotations.
For this experiment, we selected DOCS+QASCA, DOCS+MB and LCA+ME for comparison with the proposed algorithm TDH+EAI.
This is because they are the best existing algorithms for each task assignment algorithm. 
We conducted this experiment with 10 human annotators for 20 rounds on our own crowdsourcing system.
For each worker, we assigned 5 tasks in each round.
\ifrealtwocol
Figure~\ref{exp:crowd_accuracy_real},\ref{exp:crowd_accuracygen_real} and \ref{exp:crowd_avgdist_real} show the performances of the algorithms against the rounds.
\else
Figure~\ref{exp:real_crowd_bp} and \ref{exp:real_crowd_whc} show the performances of the algorithms against the rounds.
\fi
For both of the datasets, the results confirm that the proposed TDH+EAI algorithm outperforms the compared algorithms as in the previous simulations.
Without crowdsourcing, the other algorithms show a higher \accgen than TDH for \heritages dataset, because these algorithms tend to estimate the truths with more generalized form than TDH does.
However, TDH+EAI shows the highest \accgen after the 3rd round because it correctly estimates the reliabilities and the generalization levels of the sources by using the hierarchy.
For \birthplaces dataset, \accs of the algorithms increase a little bit faster than those in the experiment with simulated crowdsourcing.
However, for \heritages dataset, \accs of the algorithms increase much slower than in the experiment with simulated crowdsourcing.
It seems that finding the locations of a world heritages is a quite harder task than finding the birthplaces of celebrities because the birthplaces are often big cities (such as LA), which are familiar to workers, but World Cultural Heritages and World Natural Heritages are often located in unfamiliar regions.

\ifrealtwocol
\begin{figure}[tb]
	\centering
	\subfloat[\birthplaces]{\includegraphics[width=1.32in]{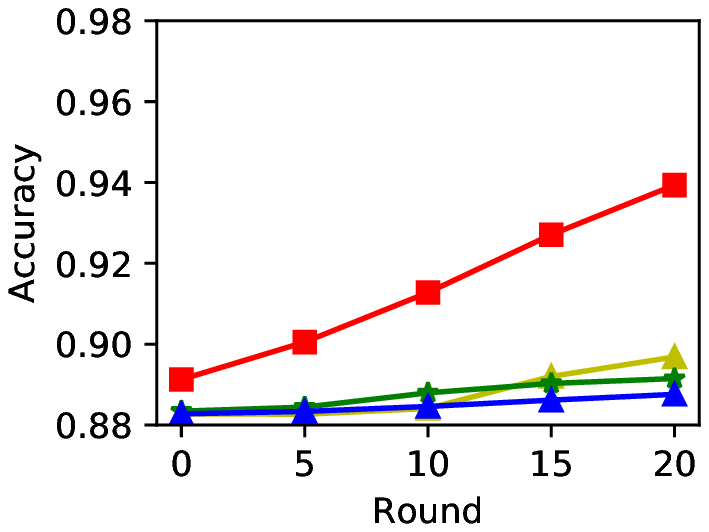}%
		\label{exp:accu_bp_real}}
	\subfloat[\heritages]{\hspace{-0.05in}\includegraphics[width=2.03in]{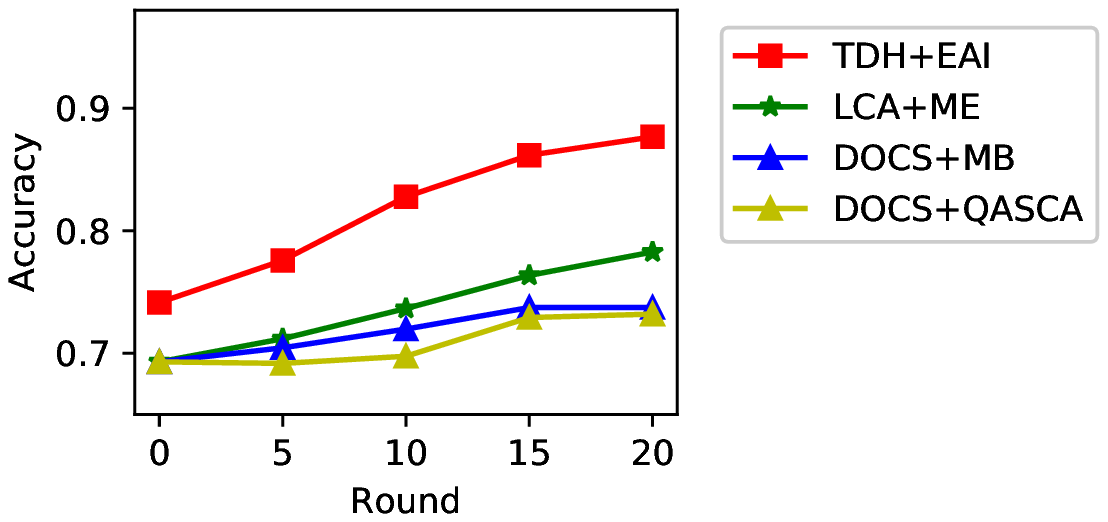}%
		\label{exp:accu_whc_real}}
	\vspace{-0.06in}
	\caption{\acc with human annotations}
	\label{exp:crowd_accuracy_real}
	\vspace{-0.09in}
\end{figure}
\begin{figure}[tb]
	\centering
	\subfloat[\birthplaces]{\includegraphics[width=1.32in]{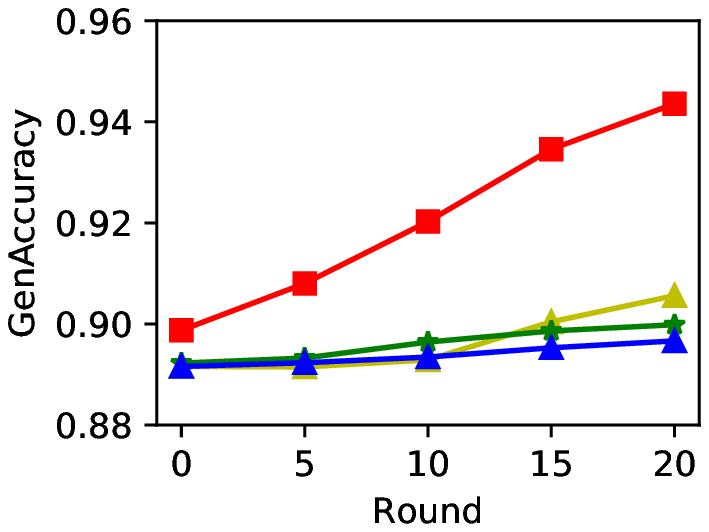}%
		\label{exp:accugen_bp_real}}
	\subfloat[\heritages]{\hspace{-0.05in}\includegraphics[width=2.03in]{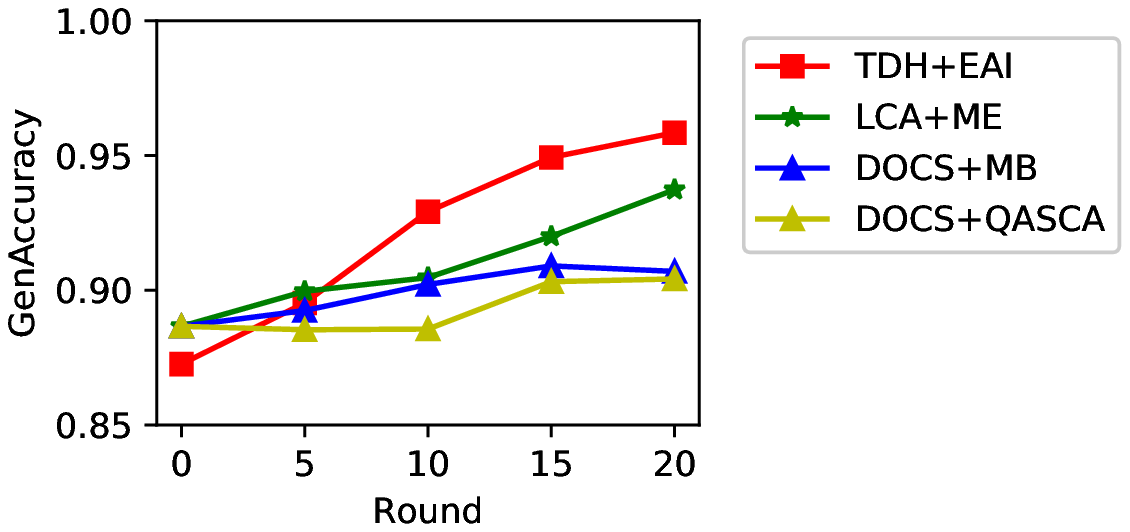}%
		\label{exp:accugen_whc_real}}
	\vspace{-0.06in}
	\caption{\accgen with human annotations}
	\label{exp:crowd_accuracygen_real}
	\vspace{-0.09in}
\end{figure}
\begin{figure}[tb]
	\centering
	\subfloat[\birthplaces]{\includegraphics[width=1.36in]{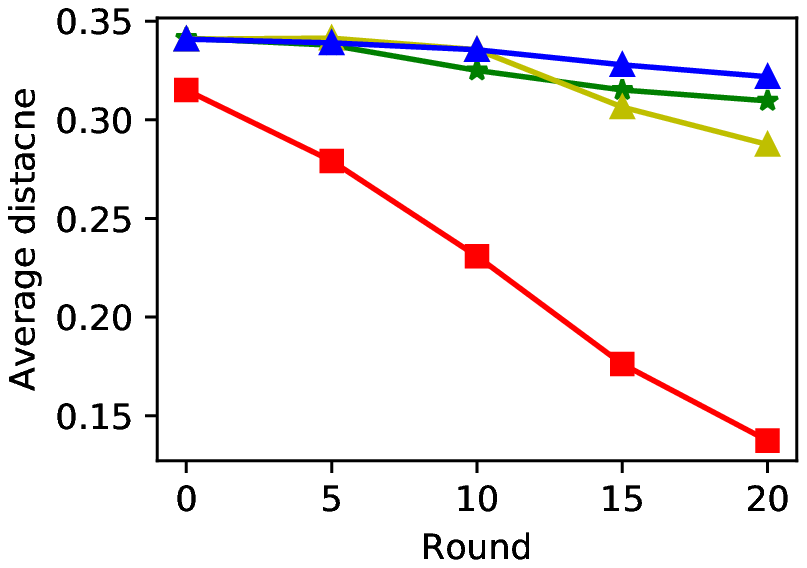}%
		\label{exp:avgdist_bp_real}}
	\subfloat[\heritages]{\hspace{-0.05in}\includegraphics[width=1.98in]{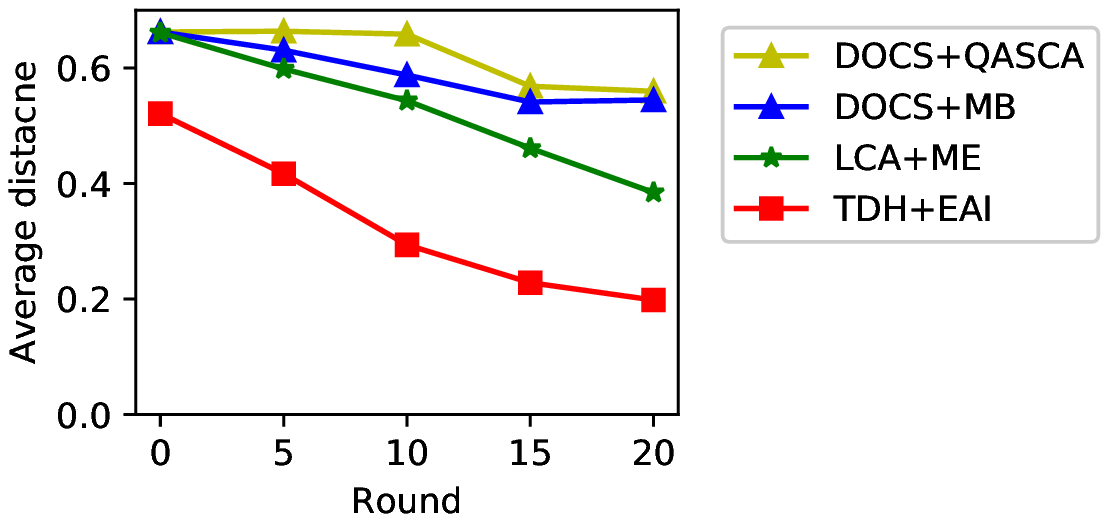}%
		\label{exp:avgdist_whc_real}}
	\caption{\avgdist with human annotations}
	\vspace{-0.06in}
	\label{exp:crowd_avgdist_real}
	\vspace{-0.09in}
\end{figure}
\else
\begin{figure}[t]
	\includegraphics[width=3.3in]{birthPlaces_Real_q5/realcrowdsourcing_round}
	\caption{Crowdsourced truth discovery in \birthplaces \label{exp:real_crowd_bp}}
	\vspace{-0.1in}
\end{figure}	
\begin{figure}[t!]
	\includegraphics[width=3.3in]{whc_Real_q5/realcrowdsourcing_round}
	\caption{Crowdsourced truth discovery in \heritages\label{exp:real_crowd_whc}}
\end{figure}	
\fi

\subsection{Crowdsourcing with AMT}
We evaluate the performances of TDH+EAI, DOCS+QASCA, DOCS+MB and LCA+ME based on the answers collected from Amazon Mechanical Turk (AMT). 
We collected answers for all objects in \heritages dataset from 20 workers in AMT.
We made the collected answers available at \url{http://kdd.snu.ac.kr/home/datasets/tdh.php}.
To evaluate the algorithms based on the collected answers, we assign 5 tasks for each worker in a round.
We plotted the performance of the algorithms in \figurename~\ref{exp:real_crowd_whc_amt}.
Since we use more workers than we did in Section~\ref{sec:expReal}, the performances improve a little bit faster, but the trends are very similar to those with 10 human annotators in the previous section. 
We observe that our TDH+EAI outperforms all compared algorithms even with a commercial crowdsourcing platform. 


\begin{figure}[t]
	\vspace{-0.08in}
	\includegraphics[width=3.3in]{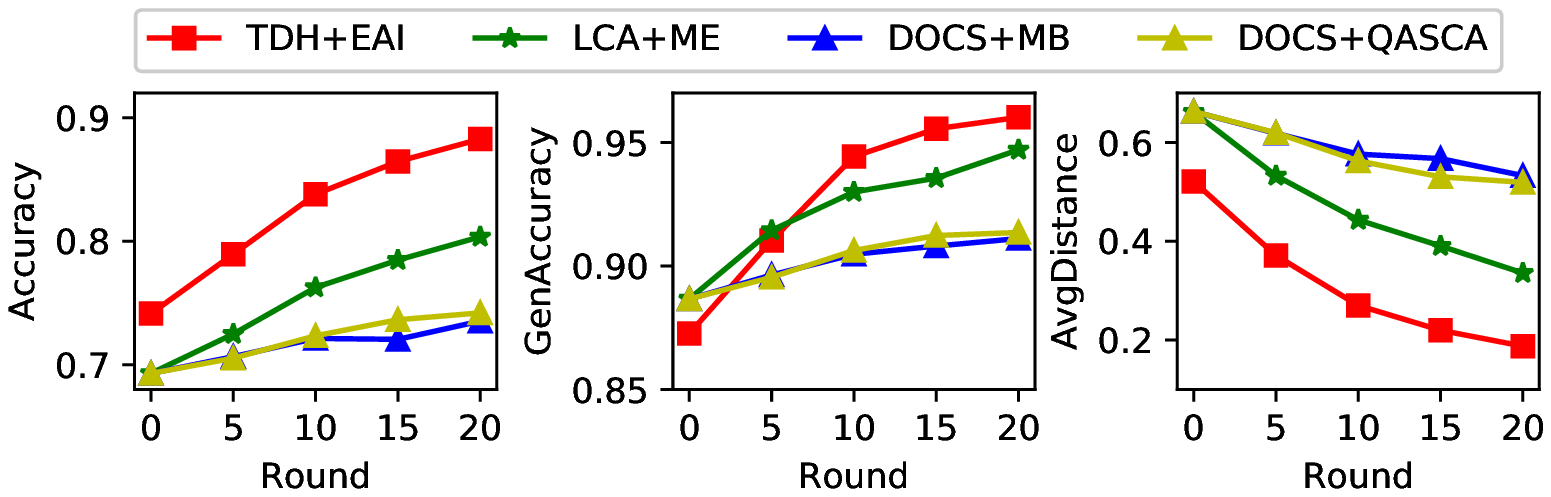}
	\vspace{-0.13in}
	\caption{Crowdsourced truth discovery in \heritages \label{exp:real_crowd_whc_amt}}
	\vspace{-0.1in}
\end{figure}
\begin{table}[b]
	\center
	\vspace{-0.06in}
	\caption{Performance of truth discovery algorithms}
	\vspace{-0.06in}
	\footnotesize
	\label{tab:mt}
	\setlength{\tabcolsep}{0.6em}
	\begin{tabular}{cc|ccc|ccc}
		\toprule
		&& \multicolumn{6}{c}{Dataset}\\
		\midrule
		&& \multicolumn{3}{c|}{\birthplaces} & \multicolumn{3}{c}{\heritages}\\
		\multicolumn{2}{c|}{Algorithm} & Precision & Recall & F1 & Precision & Recall & F1 \\
		\midrule
		\multirow{10}{*}{\makecell{Single\\truth}} & TDH&{\bf{0.899}}&{\bf{0.921}}&{\bf{0.910}}&0.873&0.795&{\bf{0.832}}\\
		& VOTE&0.892&0.804&0.846&{\bf{0.899}}&0.717&0.798\\
		& LCA&0.892&0.913&0.903&0.878&0.711&0.786\\
		& DOCS&0.892&0.913&0.902&0.887&0.722&0.796\\
		& ASUMS&0.857&0.888&0.872&0.741&0.660&0.698\\
		& POPACCU&0.847&0.858&0.852&0.859&0.694&0.768\\
		& LFC&0.874&0.838&0.856&0.808&0.727&0.765\\
		& MDC&0.844&0.853&0.848&0.807&0.792&0.800\\
		& ACCU&0.830&0.842&0.836&0.766&0.631&0.692\\
		& CRH&0.827&0.833&0.830&0.883&0.716&0.791\\
		\midrule
		\multirow{3}{*}{\makecell{Multi\\-truths}}&LFC-MT&0.763&0.723&0.742&0.898&0.684&0.777\\
		&DART&0.590&0.855&0.698&0.357&{\bf{0.994}}&0.525\\
		&LTM&0.780&0.472&0.588&0.871&0.672&0.759\\
		\bottomrule
	\end{tabular}
\end{table}
\begin{table}[b]
	\center
	\vspace{-0.06in}
	\caption{Performance evaluation for numerical data}
	\footnotesize
	\vspace{-0.06in}
	\label{tab:num}
	\setlength{\tabcolsep}{0.7em}
	\begin{tabular}{c|cc|cc|cc}
		\toprule
		& \multicolumn{2}{c|}{Change rate} & \multicolumn{2}{c}{Open price} & \multicolumn{2}{|c}{EPS}\\
		\midrule
		Algorithm& MAE & R/E & MAE & R/E & MAE & R/E\\ 
		\midrule
		TDH&{\textBF{0.0006}}&{\textBF{0.1011}}&{\textBF{0.0195}}&{\textBF{0.0354}}&{\textBF{0.0352}}&{\textBF{1.9513}}\\
		LCA&{\textBF{0.0006}}&{\textBF{0.1011}}&{\textBF{0.0195}}&{\textBF{0.0354}}&0.3831&16.2212\\
		CRH&0.0020&1.6339&{\textBF{0.0195}}&{\textBF{0.0354}}&0.0610&1.9882\\
		CATD&0.0104&2.3529&0.0211&0.0395&0.0803&3.2059\\
		VOTE&{\textBF{0.0006}}&{\textBF{0.1011}}&{\textBF{0.0195}}&{\textBF{0.0354}}&0.0765&2.8402\\
		MEAN&0.2837&30.8747&0.4047&0.5782&0.1762&7.3937\\
		\bottomrule
	\end{tabular}
\end{table}

\subsection{Multi-truths Discovery Algorithms}
Since there are multiple correct values including generalized values,
\eat{Since the truth discovery with hierarchies finds multiple correct values that are the exactly correct value as well as its generalized values,}
we also implement multi-truth discovery algorithms such as DART\cite{lin2018domain}, LFC\cite{raykar2010learning} and LTM\cite{zhao2012bayesian} to compare with our TDH algorithm.
Since the multi-truths discovery algorithms independently generate the correct values, they may output the true values where there exist a pair of true values without ancestor-descendant relationship in the hierarchy.
For example, from the given claimed values in Table~\ref{expl_records}, the multi-truth algorithms can answer that the `Statue of Liberty` is located in LA and Liberty island.
In this case, we cannot evaluate the result by our evaluation measures \acc, \accgen and \avgdist.
Thus, to evaluate the performance of the tested algorithms, we utilize precision, recall and F1-score which are the evaluation measures typically used for multi-truths discovery.
To use the multi-truths algorithms and the evaluation measures, we treat the ancestors of $v$ and $v$ itself as the multi-truths of $v$.
LFC can work as either a single truth algorithm or a multi-truths algorithm.
We refer to the multi-truth version of LFC as LFC-MT to avoid the confusion.

\tablename~\ref{tab:mt} shows the performance of the truth discovery algorithms in terms of precision, recall and F1-score. 
For both datasets, the TDH algorithm is the best in terms of F1-score.
Recall that the VOTE algorithm tends to find a generalized value of the exact truth.
Since a generalized truth generates a small number of multi-truths, the VOTE algorithm shows the highest precision in \heritages dataset.
However, since its recall is much lower than that of our TDH algorithm, the F1-score of the VOTE algorithm is lower than that of the TDH algorithm.
Similarly, although the DART algorithm has the highest recall in \heritages dataset, the precision of the DART algorithm is the smallest among the precisions of all compared algorithms.

\subsection{Performance on a Numerical Dataset}
To evaluate the extension to numerical data, we conducted an experiment on the stock datatset \cite{li2012truth}  which is trading data of 1000 stock symbols from 55 sources on every work day in July 2011. The detailed description of the data can be found in \cite{li2012truth}.
As we discussed at the end of Section~\ref{sec:emalgorithm}, we can utilize our TDH algorithm for numeric dataset with implied hierarchy.
We select three attributes `change rate', `open price' and `EPS' of the dataset, and compared our TDH algorithm with the LCA, CRH, CATD\cite{li2014confidence}, VOTE and MEAN algorithms.
Among the second best performers DOCS and LCA in Table~\ref{tab:crowdsourced50}, we use only LCA for this experiment since DOCS requires the domain information while it is not available for this dataset.
In addition to LCA, we implemented and tested the two algorithms CRH\cite{li2014resolving} and CATD\cite{li2014confidence} which are designed to find the truth in numerical data.
Recall that VOTE is a baseline algorithm which selects the candidate value collected from majority sources.
We also implemented a baseline algorithm, called MEAN, which estimates the correct value as the average of the claimed numeric values. 

Table~\ref{tab:num} shows the mean squared error (MAE) and the relative error (R/E) of the tested algorithms.
The TDH algorithm performs the best for every attribute.
The MEAN and CATD algorithms show worse performance than the other algorithms.
Since they utilize an average or a weighted average of the claimed values, they are sensitive to outliers. 
The result confirms that our TDH algorithm is effective even for numerical data.

\section{Related Work}
\label{sec:related_work}

The problem of resolving conflicts from multiple sources (i.e., truth discovery) has been extensively studied \cite{demartini2012zencrowd,dong2009integrating, dong2012less,wan2016truth,raykar2010learning,li2014resolving, li2017reliable, li2015discovery,yin2008truth,zhao2012probabilistic,DOCS, zhao2012bayesian, lin2018domain, demartini2012zencrowd,whitehill2009whose,dawid1979maximum, zhou2012learning,kim2012bayesian}.
Truth discovery for categorical data has been addressed in \cite{demartini2012zencrowd,dong2009integrating, dong2012less,  raykar2010learning,li2014resolving,DOCS,li2017reliable, yin2008truth}.
According to a recent survey \cite{zheng2017truth}, LFC\cite{raykar2010learning} and CRH\cite{li2014resolving} perform the best in an extensive experiment with the truth discovery algorithms \cite{zheng2017truth,li2014confidence,demartini2012zencrowd,whitehill2009whose,dawid1979maximum, zhou2012learning,kim2012bayesian}.
There exist other interesting algorithms \cite{dong2009integrating,dong2012less,li2017reliable,DOCS} which are not evaluated together in \cite{zheng2017truth}.
Accu\cite{dong2009integrating} and PopAccu\cite{dong2012less} combine the conflicting values extracted from different sources for the knowledge fusion \cite{KnowledgeFusion}.
They consider the dependencies between data sources to penalize the copiers' claims.
DOCS\cite{DOCS} utilizes the domain information to consider the different levels of worker expertises on various domains. 
MDC\cite{li2017reliable} is a truth discovery algorithm devised for crowdsourcing-based medical diagnosis. 
The works in \cite{li2015discovery, wan2016truth,zhao2012probabilistic} studied how to resolve conflicts in numerical data from multiple sources. 

The truth discovery algorithms in \cite{zhao2012probabilistic,zhao2012bayesian,LCA,DOCS} are based on probabilistic models. 
Resolving the conflicts in numerical data is addressed in \cite{zhao2012probabilistic} and 
discovering multiple truths for an object is studied in \cite{zhao2012bayesian}. 
Probabilistic models for finding a single truth for each object is proposed in \cite{LCA,DOCS}. 
However, none of those algorithms exploit the hierarchical relationships of claimed values for truth discovery.
The work in \cite{asums} adopts an existing algorithm to consider hierarchical relationships. 
\eat{To find the true value for each object, it greedily traverses down the hierarchy tree from the root until the confidence on the node is higher than the given thresholds $\theta$.
Thus, it
However, it requires a threshold to control the granularity of the inferred truth. 
On the contrary, our proposed algorithm automatically finds the truth without any given threshold.}

Task assignment algorithms\cite{boim2012asking,ho2013adaptive,zheng2015qasca,mavridis2016using,DOCS,fan2015icrowd} in crowdsourcing have been studied widely in recent years.
The works in \cite{boim2012asking,zheng2015qasca,DOCS} can be applied to our crowdsourced truth discovery.
For task assignment, AskIt\cite{boim2012asking} selects the most uncertain object for a worker.
Meanwhile, the task assignment algorithm in \cite{DOCS} selects the object which is expected to decrease the entropy of the confidence the most.
QASCA \cite{zheng2015qasca} chooses an object which is likely to most increase the accuracy. 
\eat{QASCA has, however, two drawbacks.
Since it computes the accuracy improvement by using a sampling-based method,
the assignment is sensitive to the sampled results.
In addition, QASCA does not consider the number of claimed values collected so far and thus the estimation may not be accurate.}
Since QASCA outperforms AskIt in the experiments presented in \cite{zheng2015qasca,DOCS}, we do not consider AskIt in our experiments.
In \cite{ho2013adaptive}, task assignment for binary classification was investigated but it is not applicable to our problem to find the correct value among multiple conflicting values.
Meanwhile, the task assignment algorithm is proposed in \cite{mavridis2016using} for the case when the required skills for each task and the skill set of every worker is available.
However, it is not applicable to our problem.
\red{A task assignment algorithm proposed in \cite{fan2015icrowd} assigns every object to a fixed number of workers.
However, since we already have claimed values from sources, we do not have to assign all objects to workers.}

\section{Conclusion}
\label{sec:conclusion}
In this paper, we first proposed a probabilistic model for truth inference to utilize the hierarchical structures in claimed values and an inference algorithm for the model.
Furthermore, we proposed an efficient algorithm to assign the tasks in crowdsourcing platforms.
The performance study with real-life datasets confirms the effectiveness of the proposed algorithms.


\begin{acks}
	We appreciate the reviewers for providing their insightful comments.
	This research was supported by Next-Generation Information Computing Development Program through the National Research Foundation of Korea(NRF) funded by the Ministry of Science, ICT (NRF-2017M3C4A7063570).
	This research was also supported by Basic Science Research Program through the National Research Foundation of Korea(NRF) funded by the Ministry of Education(NRF-2016R1D1A1A02937186).
\end{acks}

\bibliographystyle{ACM-Reference-Format}
\bibliography{ref-short}


\begin{thebibliography}{42}


\ifx \showCODEN    \undefined \def \showCODEN     #1{\unskip}     \fi
\ifx \showDOI      \undefined \def \showDOI       #1{#1}\fi
\ifx \showISBNx    \undefined \def \showISBNx     #1{\unskip}     \fi
\ifx \showISBNxiii \undefined \def \showISBNxiii  #1{\unskip}     \fi
\ifx \showISSN     \undefined \def \showISSN      #1{\unskip}     \fi
\ifx \showLCCN     \undefined \def \showLCCN      #1{\unskip}     \fi
\ifx \shownote     \undefined \def \shownote      #1{#1}          \fi
\ifx \showarticletitle \undefined \def \showarticletitle #1{#1}   \fi
\ifx \showURL      \undefined \def \showURL       {\relax}        \fi
\providecommand\bibfield[2]{#2}
\providecommand\bibinfo[2]{#2}
\providecommand\natexlab[1]{#1}
\providecommand\showeprint[2][]{arXiv:#2}

\bibitem[\protect\citeauthoryear{Auer, Bizer, Kobilarov, Lehmann, Cyganiak, and
  Ives}{Auer et~al\mbox{.}}{2007}]%
        {auer2007dbpedia}
\bibfield{author}{\bibinfo{person}{S{\"o}ren Auer}, \bibinfo{person}{Christian
  Bizer}, \bibinfo{person}{Georgi Kobilarov}, \bibinfo{person}{Jens Lehmann},
  \bibinfo{person}{Richard Cyganiak}, {and} \bibinfo{person}{Zachary Ives}.}
  \bibinfo{year}{2007}\natexlab{}.
\newblock \showarticletitle{Dbpedia: A nucleus for a web of open data}.
\newblock In \bibinfo{booktitle}{\emph{The semantic web}}.
  \bibinfo{publisher}{Springer}, \bibinfo{pages}{722--735}.
\newblock


\bibitem[\protect\citeauthoryear{Beretta, Harispe, Ranwez, and
  Mougenot}{Beretta et~al\mbox{.}}{2016}]%
        {asums}
\bibfield{author}{\bibinfo{person}{Valentina Beretta},
  \bibinfo{person}{S{\'e}bastien Harispe}, \bibinfo{person}{Sylvie Ranwez},
  {and} \bibinfo{person}{Isabelle Mougenot}.} \bibinfo{year}{2016}\natexlab{}.
\newblock \showarticletitle{How Can Ontologies Give You Clue for
  Truth-Discovery? An Exploratory Study}. In \bibinfo{booktitle}{\emph{WIMS}}.
  \bibinfo{pages}{15}.
\newblock


\bibitem[\protect\citeauthoryear{Boim, Greenshpan, Milo, Novgorodov, Polyzotis,
  and Tan}{Boim et~al\mbox{.}}{2012}]%
        {boim2012asking}
\bibfield{author}{\bibinfo{person}{Rubi Boim}, \bibinfo{person}{Ohad
  Greenshpan}, \bibinfo{person}{Tova Milo}, \bibinfo{person}{Slava Novgorodov},
  \bibinfo{person}{Neoklis Polyzotis}, {and} \bibinfo{person}{Wang-Chiew Tan}.}
  \bibinfo{year}{2012}\natexlab{}.
\newblock \showarticletitle{Asking the right questions in crowd data sourcing}.
  In \bibinfo{booktitle}{\emph{ICDE}}. \bibinfo{pages}{1261--1264}.
\newblock


\bibitem[\protect\citeauthoryear{Dawid and Skene}{Dawid and Skene}{1979}]%
        {dawid1979maximum}
\bibfield{author}{\bibinfo{person}{Alexander~Philip Dawid} {and}
  \bibinfo{person}{Allan~M Skene}.} \bibinfo{year}{1979}\natexlab{}.
\newblock \showarticletitle{Maximum likelihood estimation of observer
  error-rates using the EM algorithm}.
\newblock \bibinfo{journal}{\emph{Applied statistics}} (\bibinfo{year}{1979}),
  \bibinfo{pages}{20--28}.
\newblock


\bibitem[\protect\citeauthoryear{Demartini, Difallah, and
  Cudr{\'e}-Mauroux}{Demartini et~al\mbox{.}}{2012}]%
        {demartini2012zencrowd}
\bibfield{author}{\bibinfo{person}{Gianluca Demartini},
  \bibinfo{person}{Djellel~Eddine Difallah}, {and} \bibinfo{person}{Philippe
  Cudr{\'e}-Mauroux}.} \bibinfo{year}{2012}\natexlab{}.
\newblock \showarticletitle{ZenCrowd: leveraging probabilistic reasoning and
  crowdsourcing techniques for large-scale entity linking}. In
  \bibinfo{booktitle}{\emph{WWW}}. \bibinfo{pages}{469--478}.
\newblock


\bibitem[\protect\citeauthoryear{Dong, Gabrilovich, Heitz, Horn, Lao, Murphy,
  Strohmann, Sun, and Zhang}{Dong et~al\mbox{.}}{2014a}]%
        {KnowledgeVault}
\bibfield{author}{\bibinfo{person}{Xin Dong}, \bibinfo{person}{Evgeniy
  Gabrilovich}, \bibinfo{person}{Geremy Heitz}, \bibinfo{person}{Wilko Horn},
  \bibinfo{person}{Ni Lao}, \bibinfo{person}{Kevin Murphy},
  \bibinfo{person}{Thomas Strohmann}, \bibinfo{person}{Shaohua Sun}, {and}
  \bibinfo{person}{Wei Zhang}.} \bibinfo{year}{2014}\natexlab{a}.
\newblock \showarticletitle{Knowledge Vault: A Web-scale Approach to
  Probabilistic Knowledge Fusion}. In \bibinfo{booktitle}{\emph{SIGKDD}}.
  \bibinfo{pages}{601--610}.
\newblock
\showISBNx{978-1-4503-2956-9}


\bibitem[\protect\citeauthoryear{Dong, Berti-Equille, and Srivastava}{Dong
  et~al\mbox{.}}{2009}]%
        {dong2009integrating}
\bibfield{author}{\bibinfo{person}{Xin~Luna Dong}, \bibinfo{person}{Laure
  Berti-Equille}, {and} \bibinfo{person}{Divesh Srivastava}.}
  \bibinfo{year}{2009}\natexlab{}.
\newblock \showarticletitle{Integrating conflicting data: the role of source
  dependence}.
\newblock \bibinfo{journal}{\emph{PVLDB}} \bibinfo{volume}{2},
  \bibinfo{number}{1} (\bibinfo{year}{2009}), \bibinfo{pages}{550--561}.
\newblock


\bibitem[\protect\citeauthoryear{Dong, Gabrilovich, Heitz, Horn, Murphy, Sun,
  and Zhang}{Dong et~al\mbox{.}}{2014b}]%
        {KnowledgeFusion}
\bibfield{author}{\bibinfo{person}{Xin~Luna Dong}, \bibinfo{person}{Evgeniy
  Gabrilovich}, \bibinfo{person}{Geremy Heitz}, \bibinfo{person}{Wilko Horn},
  \bibinfo{person}{Kevin Murphy}, \bibinfo{person}{Shaohua Sun}, {and}
  \bibinfo{person}{Wei Zhang}.} \bibinfo{year}{2014}\natexlab{b}.
\newblock \showarticletitle{From data fusion to knowledge fusion}.
\newblock \bibinfo{journal}{\emph{PVLDB}} \bibinfo{volume}{7},
  \bibinfo{number}{10} (\bibinfo{year}{2014}), \bibinfo{pages}{881--892}.
\newblock
\showISSN{21508097}


\bibitem[\protect\citeauthoryear{Dong, Saha, and Srivastava}{Dong
  et~al\mbox{.}}{2012}]%
        {dong2012less}
\bibfield{author}{\bibinfo{person}{Xin~Luna Dong}, \bibinfo{person}{Barna
  Saha}, {and} \bibinfo{person}{Divesh Srivastava}.}
  \bibinfo{year}{2012}\natexlab{}.
\newblock \showarticletitle{Less is more: Selecting sources wisely for
  integration}. In \bibinfo{booktitle}{\emph{PVLDB}}, Vol.~\bibinfo{volume}{6}.
  \bibinfo{pages}{37--48}.
\newblock


\bibitem[\protect\citeauthoryear{Dong and Srivastava}{Dong and
  Srivastava}{2015}]%
        {dong2015knowledge}
\bibfield{author}{\bibinfo{person}{Xin~Luna Dong} {and} \bibinfo{person}{Divesh
  Srivastava}.} \bibinfo{year}{2015}\natexlab{}.
\newblock \showarticletitle{Knowledge curation and knowledge fusion:
  challenges, models and applications}. In \bibinfo{booktitle}{\emph{SIGMOD}}.
  \bibinfo{pages}{2063--2066}.
\newblock


\bibitem[\protect\citeauthoryear{Fan, Li, Ooi, Tan, and Feng}{Fan
  et~al\mbox{.}}{2015}]%
        {fan2015icrowd}
\bibfield{author}{\bibinfo{person}{Ju Fan}, \bibinfo{person}{Guoliang Li},
  \bibinfo{person}{Beng~Chin Ooi}, \bibinfo{person}{Kian-lee Tan}, {and}
  \bibinfo{person}{Jianhua Feng}.} \bibinfo{year}{2015}\natexlab{}.
\newblock \showarticletitle{icrowd: An adaptive crowdsourcing framework}. In
  \bibinfo{booktitle}{\emph{SIGMOD}}. \bibinfo{pages}{1015--1030}.
\newblock


\bibitem[\protect\citeauthoryear{Fan, Lu, Ooi, Tan, and Zhang}{Fan
  et~al\mbox{.}}{2014}]%
        {fan2014hybrid}
\bibfield{author}{\bibinfo{person}{Ju Fan}, \bibinfo{person}{Meiyu Lu},
  \bibinfo{person}{Beng~Chin Ooi}, \bibinfo{person}{Wang-Chiew Tan}, {and}
  \bibinfo{person}{Meihui Zhang}.} \bibinfo{year}{2014}\natexlab{}.
\newblock \showarticletitle{A hybrid machine-crowdsourcing system for matching
  web tables}. In \bibinfo{booktitle}{\emph{ICDE}}. \bibinfo{pages}{976--987}.
\newblock


\bibitem[\protect\citeauthoryear{Haas, Wang, Wu, and Franklin}{Haas
  et~al\mbox{.}}{2015}]%
        {haas2015clamshell}
\bibfield{author}{\bibinfo{person}{Daniel Haas}, \bibinfo{person}{Jiannan
  Wang}, \bibinfo{person}{Eugene Wu}, {and} \bibinfo{person}{Michael~J
  Franklin}.} \bibinfo{year}{2015}\natexlab{}.
\newblock \showarticletitle{Clamshell: Speeding up crowds for low-latency data
  labeling}.
\newblock \bibinfo{journal}{\emph{PVLDB}} \bibinfo{volume}{9},
  \bibinfo{number}{4} (\bibinfo{year}{2015}), \bibinfo{pages}{372--383}.
\newblock


\bibitem[\protect\citeauthoryear{Ho, Jabbari, and Vaughan}{Ho
  et~al\mbox{.}}{2013}]%
        {ho2013adaptive}
\bibfield{author}{\bibinfo{person}{Chien-Ju Ho}, \bibinfo{person}{Shahin
  Jabbari}, {and} \bibinfo{person}{Jennifer~Wortman Vaughan}.}
  \bibinfo{year}{2013}\natexlab{}.
\newblock \showarticletitle{Adaptive task assignment for crowdsourced
  classification}. In \bibinfo{booktitle}{\emph{ICML}}.
  \bibinfo{pages}{534--542}.
\newblock


\bibitem[\protect\citeauthoryear{Karger, Oh, and Shah}{Karger
  et~al\mbox{.}}{2011}]%
        {karger2011iterative}
\bibfield{author}{\bibinfo{person}{David~R Karger}, \bibinfo{person}{Sewoong
  Oh}, {and} \bibinfo{person}{Devavrat Shah}.} \bibinfo{year}{2011}\natexlab{}.
\newblock \showarticletitle{Iterative learning for reliable crowdsourcing
  systems}. In \bibinfo{booktitle}{\emph{NIPS}}. \bibinfo{pages}{1953--1961}.
\newblock


\bibitem[\protect\citeauthoryear{Kim and Ghahramani}{Kim and
  Ghahramani}{2012}]%
        {kim2012bayesian}
\bibfield{author}{\bibinfo{person}{Hyun-Chul Kim} {and} \bibinfo{person}{Zoubin
  Ghahramani}.} \bibinfo{year}{2012}\natexlab{}.
\newblock \showarticletitle{Bayesian classifier combination}. In
  \bibinfo{booktitle}{\emph{AISTATS}}. \bibinfo{pages}{619--627}.
\newblock


\bibitem[\protect\citeauthoryear{Kim, Jung, and Shim}{Kim
  et~al\mbox{.}}{2017a}]%
        {kim2017integration}
\bibfield{author}{\bibinfo{person}{Younghoon Kim}, \bibinfo{person}{Woohwan
  Jung}, {and} \bibinfo{person}{Kyuseok Shim}.}
  \bibinfo{year}{2017}\natexlab{a}.
\newblock \showarticletitle{Integration of graphs from different data sources
  using crowdsourcing}.
\newblock \bibinfo{journal}{\emph{Information Sciences}}  \bibinfo{volume}{385}
  (\bibinfo{year}{2017}), \bibinfo{pages}{438--456}.
\newblock


\bibitem[\protect\citeauthoryear{Kim, Kim, and Shim}{Kim
  et~al\mbox{.}}{2017b}]%
        {KimKS17}
\bibfield{author}{\bibinfo{person}{Younghoon Kim}, \bibinfo{person}{Wooyeol
  Kim}, {and} \bibinfo{person}{Kyuseok Shim}.}
  \bibinfo{year}{2017}\natexlab{b}.
\newblock \showarticletitle{Latent ranking analysis using pairwise comparisons
  in crowdsourcing platforms}.
\newblock \bibinfo{journal}{\emph{Inf. Syst.}}  \bibinfo{volume}{65}
  (\bibinfo{year}{2017}), \bibinfo{pages}{7--21}.
\newblock


\bibitem[\protect\citeauthoryear{Lewis and Gale}{Lewis and Gale}{1994}]%
        {lewis1994sequential}
\bibfield{author}{\bibinfo{person}{David~D Lewis} {and}
  \bibinfo{person}{William~A Gale}.} \bibinfo{year}{1994}\natexlab{}.
\newblock \showarticletitle{A sequential algorithm for training text
  classifiers}. In \bibinfo{booktitle}{\emph{SIGIR}}. \bibinfo{pages}{3--12}.
\newblock


\bibitem[\protect\citeauthoryear{Li, Wang, Zheng, and Franklin}{Li
  et~al\mbox{.}}{2016b}]%
        {li2016crowdsourced}
\bibfield{author}{\bibinfo{person}{Guoliang Li}, \bibinfo{person}{Jiannan
  Wang}, \bibinfo{person}{Yudian Zheng}, {and} \bibinfo{person}{Michael~J
  Franklin}.} \bibinfo{year}{2016}\natexlab{b}.
\newblock \showarticletitle{Crowdsourced data management: A survey}.
\newblock \bibinfo{journal}{\emph{TKDE}} \bibinfo{volume}{28},
  \bibinfo{number}{9} (\bibinfo{year}{2016}), \bibinfo{pages}{2296--2319}.
\newblock


\bibitem[\protect\citeauthoryear{Li, Li, Gao, Su, Zhao, Demirbas, Fan, and
  Han}{Li et~al\mbox{.}}{2014a}]%
        {li2014confidence}
\bibfield{author}{\bibinfo{person}{Qi Li}, \bibinfo{person}{Yaliang Li},
  \bibinfo{person}{Jing Gao}, \bibinfo{person}{Lu Su}, \bibinfo{person}{Bo
  Zhao}, \bibinfo{person}{Murat Demirbas}, \bibinfo{person}{Wei Fan}, {and}
  \bibinfo{person}{Jiawei Han}.} \bibinfo{year}{2014}\natexlab{a}.
\newblock \showarticletitle{A confidence-aware approach for truth discovery on
  long-tail data}.
\newblock \bibinfo{journal}{\emph{PVLDB}} \bibinfo{volume}{8},
  \bibinfo{number}{4} (\bibinfo{year}{2014}), \bibinfo{pages}{425--436}.
\newblock


\bibitem[\protect\citeauthoryear{Li, Li, Gao, Zhao, Fan, and Han}{Li
  et~al\mbox{.}}{2014b}]%
        {li2014resolving}
\bibfield{author}{\bibinfo{person}{Qi Li}, \bibinfo{person}{Yaliang Li},
  \bibinfo{person}{Jing Gao}, \bibinfo{person}{Bo Zhao}, \bibinfo{person}{Wei
  Fan}, {and} \bibinfo{person}{Jiawei Han}.} \bibinfo{year}{2014}\natexlab{b}.
\newblock \showarticletitle{Resolving conflicts in heterogeneous data by truth
  discovery and source reliability estimation}. In
  \bibinfo{booktitle}{\emph{SIGMOD}}. \bibinfo{pages}{1187--1198}.
\newblock


\bibitem[\protect\citeauthoryear{Li, Dong, Lyons, Meng, and Srivastava}{Li
  et~al\mbox{.}}{2012}]%
        {li2012truth}
\bibfield{author}{\bibinfo{person}{Xian Li}, \bibinfo{person}{Xin~Luna Dong},
  \bibinfo{person}{Kenneth Lyons}, \bibinfo{person}{Weiyi Meng}, {and}
  \bibinfo{person}{Divesh Srivastava}.} \bibinfo{year}{2012}\natexlab{}.
\newblock \showarticletitle{Truth finding on the deep web: Is the problem
  solved?}. In \bibinfo{booktitle}{\emph{PVLDB}}, Vol.~\bibinfo{volume}{6}.
  \bibinfo{pages}{97--108}.
\newblock


\bibitem[\protect\citeauthoryear{Li, Du, Liu, Xie, Fan, Li, Gao, and Sun}{Li
  et~al\mbox{.}}{2017}]%
        {li2017reliable}
\bibfield{author}{\bibinfo{person}{Yaliang Li}, \bibinfo{person}{Nan Du},
  \bibinfo{person}{Chaochun Liu}, \bibinfo{person}{Yusheng Xie},
  \bibinfo{person}{Wei Fan}, \bibinfo{person}{Qi Li}, \bibinfo{person}{Jing
  Gao}, {and} \bibinfo{person}{Huan Sun}.} \bibinfo{year}{2017}\natexlab{}.
\newblock \showarticletitle{Reliable Medical Diagnosis from Crowdsourcing:
  Discover Trustworthy Answers from Non-Experts}. In
  \bibinfo{booktitle}{\emph{WSDM}}. \bibinfo{pages}{253--261}.
\newblock


\bibitem[\protect\citeauthoryear{Li, Gao, Meng, Li, Su, Zhao, Fan, and Han}{Li
  et~al\mbox{.}}{2016a}]%
        {TDSurveryKDDExp15}
\bibfield{author}{\bibinfo{person}{Yaliang Li}, \bibinfo{person}{Jing Gao},
  \bibinfo{person}{Chuishi Meng}, \bibinfo{person}{Qi Li}, \bibinfo{person}{Lu
  Su}, \bibinfo{person}{Bo Zhao}, \bibinfo{person}{Wei Fan}, {and}
  \bibinfo{person}{Jiawei Han}.} \bibinfo{year}{2016}\natexlab{a}.
\newblock \showarticletitle{A Survey on Truth Discovery}.
\newblock \bibinfo{journal}{\emph{SIGKDD Explor. Newsl.}} \bibinfo{volume}{17},
  \bibinfo{number}{2} (\bibinfo{date}{Feb.} \bibinfo{year}{2016}),
  \bibinfo{pages}{1--16}.
\newblock
\showISSN{1931-0145}


\bibitem[\protect\citeauthoryear{Li, Li, Gao, Su, Zhao, Fan, and Han}{Li
  et~al\mbox{.}}{2015}]%
        {li2015discovery}
\bibfield{author}{\bibinfo{person}{Yaliang Li}, \bibinfo{person}{Qi Li},
  \bibinfo{person}{Jing Gao}, \bibinfo{person}{Lu Su}, \bibinfo{person}{Bo
  Zhao}, \bibinfo{person}{Wei Fan}, {and} \bibinfo{person}{Jiawei Han}.}
  \bibinfo{year}{2015}\natexlab{}.
\newblock \showarticletitle{On the discovery of evolving truth}. In
  \bibinfo{booktitle}{\emph{SIGKDD}}. \bibinfo{pages}{675--684}.
\newblock


\bibitem[\protect\citeauthoryear{Lin and Chen}{Lin and Chen}{2018}]%
        {lin2018domain}
\bibfield{author}{\bibinfo{person}{Xueling Lin} {and} \bibinfo{person}{Lei
  Chen}.} \bibinfo{year}{2018}\natexlab{}.
\newblock \showarticletitle{Domain-aware multi-truth discovery from conflicting
  sources}.
\newblock \bibinfo{journal}{\emph{PVLDB}} \bibinfo{volume}{11},
  \bibinfo{number}{5} (\bibinfo{year}{2018}), \bibinfo{pages}{635--647}.
\newblock


\bibitem[\protect\citeauthoryear{Mavridis, Gross-Amblard, and
  Mikl{\'o}s}{Mavridis et~al\mbox{.}}{2016}]%
        {mavridis2016using}
\bibfield{author}{\bibinfo{person}{Panagiotis Mavridis}, \bibinfo{person}{David
  Gross-Amblard}, {and} \bibinfo{person}{Zolt{\'a}n Mikl{\'o}s}.}
  \bibinfo{year}{2016}\natexlab{}.
\newblock \showarticletitle{Using hierarchical skills for optimized task
  assignment in knowledge-intensive crowdsourcing}. In
  \bibinfo{booktitle}{\emph{WWW}}. \bibinfo{pages}{843--853}.
\newblock


\bibitem[\protect\citeauthoryear{Pasternack and Roth}{Pasternack and
  Roth}{2010}]%
        {sums}
\bibfield{author}{\bibinfo{person}{Jeff Pasternack} {and} \bibinfo{person}{Dan
  Roth}.} \bibinfo{year}{2010}\natexlab{}.
\newblock \showarticletitle{Knowing what to believe (when you already know
  something)}. In \bibinfo{booktitle}{\emph{COLING}}.
  \bibinfo{pages}{877--885}.
\newblock


\bibitem[\protect\citeauthoryear{Pasternack and Roth}{Pasternack and
  Roth}{2013}]%
        {LCA}
\bibfield{author}{\bibinfo{person}{Jeff Pasternack} {and} \bibinfo{person}{Dan
  Roth}.} \bibinfo{year}{2013}\natexlab{}.
\newblock \showarticletitle{Latent Credibility Analysis}. In
  \bibinfo{booktitle}{\emph{WWW}}. \bibinfo{pages}{1009--1020}.
\newblock


\bibitem[\protect\citeauthoryear{Raykar, Yu, Zhao, Valadez, Florin, Bogoni, and
  Moy}{Raykar et~al\mbox{.}}{2010}]%
        {raykar2010learning}
\bibfield{author}{\bibinfo{person}{Vikas~C Raykar}, \bibinfo{person}{Shipeng
  Yu}, \bibinfo{person}{Linda~H Zhao}, \bibinfo{person}{Gerardo~Hermosillo
  Valadez}, \bibinfo{person}{Charles Florin}, \bibinfo{person}{Luca Bogoni},
  {and} \bibinfo{person}{Linda Moy}.} \bibinfo{year}{2010}\natexlab{}.
\newblock \showarticletitle{Learning from crowds}.
\newblock \bibinfo{journal}{\emph{JMLR}} \bibinfo{volume}{11},
  \bibinfo{number}{Apr} (\bibinfo{year}{2010}), \bibinfo{pages}{1297--1322}.
\newblock


\bibitem[\protect\citeauthoryear{University}{University}{2010}]%
        {wordnet}
\bibfield{author}{\bibinfo{person}{Princeton University}.}
  \bibinfo{year}{2010}\natexlab{}.
\newblock \bibinfo{title}{About WordNet.}
\newblock
\newblock
\urldef\tempurl%
\url{https://wordnet.princeton.edu/}
\showURL{%
\tempurl}


\bibitem[\protect\citeauthoryear{Wan, Chen, Kaplan, Han, Gao, and Zhao}{Wan
  et~al\mbox{.}}{2016}]%
        {wan2016truth}
\bibfield{author}{\bibinfo{person}{Mengting Wan}, \bibinfo{person}{Xiangyu
  Chen}, \bibinfo{person}{Lance Kaplan}, \bibinfo{person}{Jiawei Han},
  \bibinfo{person}{Jing Gao}, {and} \bibinfo{person}{Bo Zhao}.}
  \bibinfo{year}{2016}\natexlab{}.
\newblock \showarticletitle{From Truth Discovery to Trustworthy Opinion
  Discovery: An Uncertainty-Aware Quantitative Modeling Approach}. In
  \bibinfo{booktitle}{\emph{SIGKDD}}. \bibinfo{pages}{1885--1894}.
\newblock


\bibitem[\protect\citeauthoryear{Wang, Kraska, Franklin, and Feng}{Wang
  et~al\mbox{.}}{2012}]%
        {wang2012crowder}
\bibfield{author}{\bibinfo{person}{Jiannan Wang}, \bibinfo{person}{Tim Kraska},
  \bibinfo{person}{Michael~J Franklin}, {and} \bibinfo{person}{Jianhua Feng}.}
  \bibinfo{year}{2012}\natexlab{}.
\newblock \showarticletitle{Crowder: Crowdsourcing entity resolution}.
\newblock \bibinfo{journal}{\emph{PVLDB}} \bibinfo{volume}{5},
  \bibinfo{number}{11} (\bibinfo{year}{2012}), \bibinfo{pages}{1483--1494}.
\newblock


\bibitem[\protect\citeauthoryear{Whitehill, Wu, Bergsma, Movellan, and
  Ruvolo}{Whitehill et~al\mbox{.}}{2009}]%
        {whitehill2009whose}
\bibfield{author}{\bibinfo{person}{Jacob Whitehill}, \bibinfo{person}{Ting-fan
  Wu}, \bibinfo{person}{Jacob Bergsma}, \bibinfo{person}{Javier~R Movellan},
  {and} \bibinfo{person}{Paul~L Ruvolo}.} \bibinfo{year}{2009}\natexlab{}.
\newblock \showarticletitle{Whose vote should count more: Optimal integration
  of labels from labelers of unknown expertise}. In
  \bibinfo{booktitle}{\emph{NIPS}}. \bibinfo{pages}{2035--2043}.
\newblock


\bibitem[\protect\citeauthoryear{Yin, Han, and Philip}{Yin
  et~al\mbox{.}}{2008}]%
        {yin2008truth}
\bibfield{author}{\bibinfo{person}{Xiaoxin Yin}, \bibinfo{person}{Jiawei Han},
  {and} \bibinfo{person}{S~Yu Philip}.} \bibinfo{year}{2008}\natexlab{}.
\newblock \showarticletitle{Truth discovery with multiple conflicting
  information providers on the web}.
\newblock \bibinfo{journal}{\emph{TKDE}} \bibinfo{volume}{20},
  \bibinfo{number}{6} (\bibinfo{year}{2008}), \bibinfo{pages}{796--808}.
\newblock


\bibitem[\protect\citeauthoryear{Zhao and Han}{Zhao and Han}{2012}]%
        {zhao2012probabilistic}
\bibfield{author}{\bibinfo{person}{Bo Zhao} {and} \bibinfo{person}{Jiawei
  Han}.} \bibinfo{year}{2012}\natexlab{}.
\newblock \showarticletitle{A probabilistic model for estimating real-valued
  truth from conflicting sources}. In \bibinfo{booktitle}{\emph{QDB}}.
\newblock


\bibitem[\protect\citeauthoryear{Zhao, Rubinstein, Gemmell, and Han}{Zhao
  et~al\mbox{.}}{2012}]%
        {zhao2012bayesian}
\bibfield{author}{\bibinfo{person}{Bo Zhao}, \bibinfo{person}{Benjamin~IP
  Rubinstein}, \bibinfo{person}{Jim Gemmell}, {and} \bibinfo{person}{Jiawei
  Han}.} \bibinfo{year}{2012}\natexlab{}.
\newblock \showarticletitle{A bayesian approach to discovering truth from
  conflicting sources for data integration}.
\newblock \bibinfo{journal}{\emph{PVLDB}} \bibinfo{volume}{5},
  \bibinfo{number}{6} (\bibinfo{year}{2012}), \bibinfo{pages}{550--561}.
\newblock


\bibitem[\protect\citeauthoryear{Zheng, Li, and Cheng}{Zheng
  et~al\mbox{.}}{2016}]%
        {DOCS}
\bibfield{author}{\bibinfo{person}{Yudian Zheng}, \bibinfo{person}{Guoliang
  Li}, {and} \bibinfo{person}{Reynold Cheng}.} \bibinfo{year}{2016}\natexlab{}.
\newblock \showarticletitle{DOCS: a domain-aware crowdsourcing system using
  knowledge bases}.
\newblock \bibinfo{journal}{\emph{PVLDB}} \bibinfo{volume}{10},
  \bibinfo{number}{4} (\bibinfo{year}{2016}), \bibinfo{pages}{361--372}.
\newblock


\bibitem[\protect\citeauthoryear{Zheng, Li, Li, Shan, and Cheng}{Zheng
  et~al\mbox{.}}{2017}]%
        {zheng2017truth}
\bibfield{author}{\bibinfo{person}{Yudian Zheng}, \bibinfo{person}{Guoliang
  Li}, \bibinfo{person}{Yuanbing Li}, \bibinfo{person}{Caihua Shan}, {and}
  \bibinfo{person}{Reynold Cheng}.} \bibinfo{year}{2017}\natexlab{}.
\newblock \showarticletitle{Truth inference in crowdsourcing: is the problem
  solved?}
\newblock \bibinfo{journal}{\emph{PVLDB}} \bibinfo{volume}{10},
  \bibinfo{number}{5} (\bibinfo{year}{2017}), \bibinfo{pages}{541--552}.
\newblock


\bibitem[\protect\citeauthoryear{Zheng, Wang, Li, Cheng, and Feng}{Zheng
  et~al\mbox{.}}{2015}]%
        {zheng2015qasca}
\bibfield{author}{\bibinfo{person}{Yudian Zheng}, \bibinfo{person}{Jiannan
  Wang}, \bibinfo{person}{Guoliang Li}, \bibinfo{person}{Reynold Cheng}, {and}
  \bibinfo{person}{Jianhua Feng}.} \bibinfo{year}{2015}\natexlab{}.
\newblock \showarticletitle{QASCA: A quality-aware task assignment system for
  crowdsourcing applications}. In \bibinfo{booktitle}{\emph{SIGMOD}}.
  \bibinfo{pages}{1031--1046}.
\newblock


\bibitem[\protect\citeauthoryear{Zhou, Basu, Mao, and Platt}{Zhou
  et~al\mbox{.}}{2012}]%
        {zhou2012learning}
\bibfield{author}{\bibinfo{person}{Denny Zhou}, \bibinfo{person}{Sumit Basu},
  \bibinfo{person}{Yi Mao}, {and} \bibinfo{person}{John~C Platt}.}
  \bibinfo{year}{2012}\natexlab{}.
\newblock \showarticletitle{Learning from the wisdom of crowds by minimax
  entropy}. In \bibinfo{booktitle}{\emph{NIPS}}. \bibinfo{pages}{2195--2203}.
\newblock


\end{thebibliography}
\end{document}